\pdfoutput=1

\documentclass[a4paper]{scrartcl}

\usepackage[english]{babel}
\usepackage[utf8]{inputenc}
\usepackage{cmap}
\usepackage[T1]{fontenc}
\usepackage{textcomp}
\usepackage{lmodern}
\usepackage{xspace}
\usepackage{enumitem}
\usepackage{amsmath}

\usepackage{braket}

\usepackage[babel,autostyle=true,strict]{csquotes}
\MakeOuterQuote{"}

\usepackage{mathtools}

\usepackage{fixmath}

\usepackage{microtype}
\usepackage[xspace]{ellipsis}

\usepackage{stmaryrd}
\SetSymbolFont{stmry}{bold}{U}{stmry}{m}{n}

\usepackage{stackengine}

\usepackage[lined,boxed,commentsnumbered,ruled,vlined,linesnumbered]{algorithm2e}

\usepackage{booktabs}
\usepackage{colortbl}
\usepackage{makecell}
\usepackage{array}

\usepackage{tikz}
\usetikzlibrary{arrows}
\usetikzlibrary{positioning}
\usetikzlibrary{backgrounds}
\usetikzlibrary{patterns,decorations.pathreplacing}

\allowdisplaybreaks

\usepackage{amsthm}
\usepackage{amssymb}

\usepackage{url}
\usepackage{hyperref}

\definecolor{pdfanchorcolor}{named}{black}
\definecolor{pdfmenucolor}{named}{red}
\definecolor{pdfruncolor}{named}{cyan}
\definecolor{pdfurlcolor}{named}{magenta}
\definecolor{pdffilecolor}{rgb}{0.7,0,0}
\definecolor{pdflinkcolor}{rgb}{1,0,0}
\definecolor{pdfcitecolor}{rgb}{0,0.5,0}

\hypersetup{debug=false,hypertexnames=true,naturalnames=false,setpagesize=true,raiselinks=true,breaklinks=true,pageanchor=true,plainpages=false,linktocpage=true,colorlinks=true}\hypersetup{linkcolor  =pdflinkcolor,anchorcolor=pdfanchorcolor,citecolor  =pdfcitecolor,filecolor  =pdffilecolor,menucolor  =pdfmenucolor,runcolor   =pdfruncolor,urlcolor   =pdfurlcolor}
\hypersetup{bookmarksopen=true,bookmarksopenlevel=2,bookmarksnumbered=true,bookmarkstype=toc,pdfpagemode=UseOutlines,pdfstartpage=1,pdfstartview=FitV,pdfremotestartview=Fit,pdfcenterwindow=false,pdffitwindow=false,pdfnewwindow=false,pdfdisplaydoctitle=true}

\usepackage[
  style=alphabetic,minalphanames=2, maxalphanames=2,         backend=biber,
  bibwarn=true,       texencoding=auto,   bibencoding=auto,   maxnames=3,
  maxbibnames=99,
    maxcitenames=2,
  citetracker=true,
  url=false,
  isbn=false,
  doi=false,
  uniquelist
]{biblatex}

\DeclareFieldFormat[article]{title}{\mkbibemph{#1}}
\DeclareFieldFormat[article]{journal}{#1}
\DeclareFieldFormat[article]{journaltitle}{#1}
\DeclareFieldFormat[article]{volume}{\textbf{#1}}
\DeclareFieldFormat[article]{number}{no.~#1}

\DeclareFieldFormat[incollection]{title}{\mkbibemph{#1}}
\DeclareFieldFormat[incollection]{booktitle}{#1}
\DeclareFieldFormat[incollection]{volume}{\textbf{#1}}
\DeclareFieldFormat[incollection]{number}{(#1)}

\DeclareFieldFormat[inproceedings]{title}{\mkbibemph{#1}}
\DeclareFieldFormat[inproceedings]{booktitle}{#1}
\DeclareFieldFormat[inproceedings]{volume}{\textbf{#1}}
\DeclareFieldFormat[inproceedings]{number}{(#1)}

\DeclareFieldFormat[unpublished]{title}{\mkbibemph{#1}}

\DeclareFieldFormat[thesis]{title}{\mkbibemph{#1}}

\renewbibmacro*{volume+number+eid}{  \printfield{volume}\space    \printtext[parens]{\usebibmacro{date}}
  \setunit{\addcomma\space}
  \printfield{number}  \setunit{\addcomma\space}  \printfield{eid}
  }

\renewbibmacro*{issue+date}{  \iffieldundef{issue}
      {}
      {\printfield{issue}}  \newunit}

\renewbibmacro{in:}{}

\bibliography{quirk}

\newcommand{\eg}{e.g.\@\xspace}
\newcommand{\wrt}{w.\,r.\,t.\@\xspace}
\newcommand{\Wloss}{W.l.o.g.\@\xspace}
\newcommand{\wloss}{w.l.o.g.\@\xspace}

\providecommand{\dfn}{\mathrel{\mathop:}=}
\providecommand{\ddfn}{\mathrel{\mathop{{\mathop:}{\mathop:}}}=}
\newcommand{\arity}[1]{{\protect\ensuremath{\mathCommandFont{ar}(#1)}}}
\newcommand{\logicClFont}[1]{\mathcal{#1}}
\newcommand{\logicOpFont}[1]{\mathsf{#1}}
\newcommand{\problemFont}[1]{\mathrm{#1}}
\newcommand{\mathCommandFont}[1]{\mathrm{#1}}
\newcommand{\complClFont}[1]{\mathbf{#1}}
\newcommand{\td}{\mathrm{td}}

\newcommand{\PS}{\protect\ensuremath\logicClFont{PS}}
\newcommand{\CloneBF}{\protect\ensuremath{\logicOpFont{BF}}}
\newcommand{\CloneSE}{\protect\ensuremath{\logicOpFont{S_1}}}
\newcommand{\TL}{\protect\ensuremath{\logicOpFont{TL}}}
\newcommand{\X}{\protect\ensuremath\logicOpFont{X}}
\newcommand{\F}{\protect\ensuremath\logicOpFont{F}}
\newcommand{\G}{\protect\ensuremath\logicOpFont{G}}
\newcommand{\A}{\protect\ensuremath\logicOpFont{A}}
\newcommand{\E}{\protect\ensuremath\logicOpFont{E}}
\newcommand{\U}{\protect\ensuremath\logicOpFont{U}}
\newcommand{\RLS}{\protect\ensuremath\logicOpFont{R}}
\newcommand{\B}{\protect\ensuremath{\logicClFont{B}}\xspace}
\newcommand{\AU}{\A\U}
\newcommand{\AX}{\A\X}
\newcommand{\AF}{\A\F}
\newcommand{\AG}{\A\G}
\newcommand{\EX}{\E\X}
\newcommand{\EF}{\E\F}
\newcommand{\EG}{\E\G}
\newcommand{\EU}{\E\U}
\newcommand{\calB}{\protect\ensuremath{\mathcal{B}}}
\newcommand{\calM}{\protect\ensuremath{\mathcal{M}}}
\newcommand{\calQ}{\protect\ensuremath{\mathcal{Q}}}

\newcommand{\calT}{\protect\ensuremath{\mathcal{T}}}

\newcommand{\powerset}[1]{{\protect\ensuremath{\mathord{\mathfrak{P}(\nobreak#1\nobreak)}}}}
\newcommand{\SAT}{\protect\ensuremath\problemFont{SAT}\xspace}
\newcommand{\size}[1]{{\protect\ensuremath{\vert\nobreak#1\nobreak\vert}}}
\newcommand{\N}{\protect\ensuremath{\mathbb{N}}\xspace}
\newcommand{\bigO}[1]{\protect\ensuremath{{\mathcal{O}(#1)}}}

\newcommand{\bigOmega}[1]{\protect\ensuremath{\Omega(#1)}}

\newcommand{\bigTheta}[1]{\protect\ensuremath{{\Theta(#1)}}}
\newcommand{\leqlogm}{\protect\ensuremath{\leq^\mathCommandFont{log}_\mathCommandFont{m}}}

\renewcommand{\P}{\protect\ensuremath{\complClFont{P}}\xspace}
\newcommand{\EXP}{\protect\ensuremath{\complClFont{EXP}}\xspace}
\newcommand{\PSPACE}{\protect\ensuremath{\complClFont{PSPACE}}\xspace}
\newcommand{\NPSPACE}{\protect\ensuremath{\complClFont{NPSPACE}}\xspace}
\newcommand{\APSPACE}{\protect\ensuremath{\complClFont{APSPACE}}\xspace}
\newcommand{\NP}{\protect\ensuremath{\complClFont{NP}}\xspace}
\newcommand{\sfK}{\protect\ensuremath{\mathsf{K}}}

\newcommand{\sfS}{\protect\ensuremath{\mathsf{S}}}
\newcommand{\imp}{\protect\ensuremath{\rightarrow}}
\newcommand{\nimp}{\protect\ensuremath{\mathord{\nrightarrow}}}
\newcommand{\sfD}{\protect\ensuremath{\mathsf{D}}}
\newcommand{\KD}{\protect\ensuremath{\mathsf{KD}}}
\newcommand{\negg}{{\sim}}
\newcommand{\calC}{\protect\ensuremath{\mathcal{C}}}

\newcommand{\calF}{\protect\ensuremath{\mathcal{F}}}
\newcommand{\calK}{\protect\ensuremath{\mathcal{K}}}
\newcommand{\calX}{\protect\ensuremath{\mathcal{X}}}
\newcommand{\calG}{\protect\ensuremath{\mathcal{G}}}
\newcommand{\SF}[1]{\mathrm{SF}(#1)}
\newcommand{\TQBF}{\mathrm{TQBF}}

\newcommand{\ie}{i.e.\@\xspace}
\newcommand{\sland}{\land}
\newcommand{\slor}{\lor}
\DeclareMathOperator*{\bigand}{\bigwedge}
\DeclareMathOperator*{\bigor}{\bigvee}
\DeclarePairedDelimiter{\ceil}{\lceil}{\rceil}

\newcommand{\paths}{\Pi}

\newcommand{\citeref}[2]{\cite[#2]{#1}}

\newtheoremstyle{theorem}
{\bigskipamount}{\medskipamount}{}{}{\bfseries}{.}{0.5em}{}

\newtheoremstyle{example}
{\bigskipamount}{\medskipamount}{}{}{\bfseries}{:}{\newline }{}
\newtheoremstyle{remark2}
{\bigskipamount}{\medskipamount}{}{}{\itshape}{:}{0.5em}{}

\theoremstyle{plain}

\newtheorem{theorem}{Theorem}
\newtheorem*{theorem*}{Theorem}
\newtheorem{lemma}[theorem]{Lemma}
\newtheorem{corollary}[theorem]{Corollary}
\newtheorem{proposition}[theorem]{Proposition}

\theoremstyle{definition}
\newtheorem{definition}[theorem]{Definition}

\SetKwProg{Procedure}{Procedure}{}{}
  
\definecolor{np}{gray}{0.98}
\definecolor{pspace}{gray}{0.95}
\definecolor{exp}{gray}{0.9}

\definecolor{sublin}{gray}{0.98}
\definecolor{poly}{gray}{0.95}

\newcommand{\tikzcircle}[2][red,fill=red]{\tikz[baseline=-0.5ex]\draw[#1,radius=#2]
(0,0) circle ;}
\newcommand{\normalbraces}[1]{{\normalfont(}#1{\normalfont)}}

\newcommand{\phiinit}{\ensuremath\varphi_\text{init}}
\newcommand{\phiconf}{\ensuremath\varphi_\text{conf}}

\newcommand{\phikeep}[2]{\ensuremath\varphi^{(#1,#2)}_\text{keep}}

\newcommand{\phinexti}[4]{\ensuremath\varphi^{(#1,#2,#3,#4)}_\text{next}}

\begin{document}

\title{Quirky Quantifiers: Optimal Models and\\ Complexity of Computation Tree Logic}

\author{\large Martin Lück\\\large Leibniz Universität Hannover, Germany\\ \large \texttt{lueck@thi.uni-hannover.de}}
\date{\vspace{-5ex}}

\maketitle

\begin{abstract}\textbf{Abstract.}
The satisfiability problem of the branching time logic CTL is studied in terms of computational complexity. Tight upper and lower bounds are provided for each temporal operator fragment. In parallel, the minimal model size is studied with a suitable notion of minimality. Thirdly, flat CTL is investigated, \ie, formulas with very low temporal operator nesting depth. A sharp dichotomy is shown in terms of complexity and minimal models: Temporal depth one has low expressive power, while temporal depth two is equivalent to full CTL.
\end{abstract}

\section{Introduction}

\noindent\textbf{Background.} In the last decades, temporal logics
have become established as a well-known framework for verification of dynamic, reactive systems.
The first to systematically introduce time into modern logic was
Arthur Prior, who used the framework of modal logic \cite{prior57}.
The resulting language was called \emph{tense logic}. Amir Pnueli
discovered the usefulness of such logics for formally describing the behavior of dynamic systems
with discrete time steps \cite{pnueli_temporal_1977}. His suggested
method of temporal reasoning on programs
evolved into a broad family of logics; especially the linear time
logic LTL, the branching time logic CTL, and their extensions have
remarkable importance in industrial-scale software verification. They have been
researched thoroughly in terms of their expressivity and computational complexity.
In particular, the tractable model checking problem of CTL allows the application
in practice, while its satisfiability problem
is EXP-complete and therefore highly intractable
\cite{emersonTemporal,fischer_propositional_1979,Pratt1980231}.

Many \emph{fragments} of temporal logic have been investigated in the hope to find efficient algorithms. This includes restricted temporal operator sets, bounded operator nesting depth, bounded numbers of variables, and restricted sets of logical connectives \cite{demri_complexity_2002,HalpernRestrict,TLPaper,schnoeb,SC85}.
The results are not too optimistic: For no fragment of CTL or LTL the
satisfiability problem becomes tractable, except for trivial combinations of Boolean
connectives \cite{TLPaper}.
Restricting the CTL or LTL operators or the number of propositions
does not decrease the computational complexity noteworthily; and also very low
temporal depth already carries the complexity of LTL beyond
that of propositional logic \cite{demri_complexity_2002,schnoeb,SC85}.

Conversely, this means that even "simple" and "flat"
temporal formulas have sufficient expressive power, a fact
that is reflected by their application in practice. Many important
properties of computations like safety, deadlock-freeness or
fairness are expressible in temporal depth two or three. Exceptions
are CTL model checking, which is inherently sequential only for unbounded temporal
depth \cite{ctl_sequential}; furthermore modal satisfiability (as a sublogic of CTL) drops
down to NP for bounded depth, but is otherwise PSPACE-complete even
for only one proposition \cite{HalpernRestrict}.

The minimal model size of a formula---or of a class of
formulas---can serve as an indicator for its expressive power. Minimal
models are also useful to consider for algorithms that search a space of potential models, as then the
size, or other measures, of minimal models can deliver an upper bound for
the required time or memory of the algorithm. For the
fragments of CTL investigated  here, the results range over exponentially deep models, large but shallow tree-like models down to polynomial models.

The complexity of a logical satisfiability problem heavily depends on the provided set of Boolean connectives. Any finite set $C$ of Boolean connectives, which may contain functions like $\oplus$, $\rightarrow$, etc. instead of the standard connectives $\land$, $\lor$, $\neg$, forms the base of a so-called \emph{clone} $[C]$, roughly speaking the set of all Boolean functions expressible via connectives of $C$. There is a countable infinite number of distinct clones, and they form a lattice with respect to inclusion. Today it is commonly known as \emph{Post's lattice} \cite{Post41}. For a complete illustration and a list of all bases, see \eg Böhler et al.\ \cite{bohler2003playing}.

\bigskip

\noindent\textbf{Contribution.} This paper continues the systematic
study of fragments of temporal logic. We consider sublogics of CTL obtained by
limiting temporal operators, their nesting depth, or both. For each
resulting fragment, upper and lower bounds are established in terms of
computational complexity.
The notion of minimal models is introduced and upper and lower bounds are achieved,
again as a mostly complete classification of all fragments.

There are upper bounds in complexity that are corollaries from
small minimal models (like the NP cases), but several hardness results as well
yield formulas that require large models.
For this reason, it may be not surprising that the results in both dimensions
closely correlate; specifically a temporal depth of two seems to
be the ``magical threshold'' for the hardness of CTL, a behavior
that can also be observed for LTL \cite{demri_complexity_2002}.

All established upper bounds in terms of computational complexity and minimal models are clone-independent, \ie, they hold for arbitrary sets of provided Boolean connectives. The lower bounds, on the other hand, are shown for all clones that contain the \emph{negated implication} $x \nimp y$ (\ie, $x \land \neg y$). This is a consequent continuation of the work of Lewis \cite{lewis79}, who showed that propositional satisfiability over $\nimp$ is already NP-complete, whereas it is in P for all sets of Boolean connectives that are unable to express $\nimp$. In the setting of temporal logic, similarly the tractable Boolean fragments were investigated by Meier et al.\ \cite{TLPaper}.

\smallskip

The article is organized as follows. Preliminary definitions of complexity theory and temporal logic are given in Section~\ref{sec:prelim}. The main part, Section~\ref{sec:ctl-fragments}, classifies all fragments of CTL regarding the allowed temporal operators. The PSPACE-complete fragments of CTL are investigated in Subsection~\ref{sec:af} ($\AF$), \ref{sec:ag} ($\AG$), \ref{sec:ax} ($\AX$) and \ref{sec:afax} ($\AF,\AX$). The remaining fragments of CTL are all EXP-complete and are addressed in Subsection~\ref{sec:hard}. The respective subsections contain model-theoretical upper and lower bounds as well.

In contrast to the above results, Section~\ref{sec:flat-ctl} focuses on \emph{flat CTL}, \ie, all of the above fragments with temporal depth at most one. It is shown that these fragments are all NP-complete due to a polynomial model property. Finally, several meta-results with respect to Boolean clones are given in Section~\ref{sec:clones}, stating how to transfer upper and lower bounds (in the computational or in the model-theoretical sense) to different sets of Boolean connectives.

\section{Preliminaries}\label{sec:prelim}

Common mathematical symbols are used with the following meaning. $\N$ is the set of natural numbers including zero, that is, $\N \dfn \{0, 1, \ldots\}$. We however write $[n]$ for the set $\{1, \ldots, n\}$. The logarithm $\log x$ is defined to the base 2, and is  usually rounded up when mapping to the natural numbers. If nothing else is stated, consequently $\log x$ is a shorthand for $\ceil{\log_2 x}$. For base $e$, we instead write $\ln x$.

\subsection*{Complexity theory}

Using the standard concept of resource-bounded Turing machines, we refer to common complexity classes as follows. A computational problem is included in
\begin{itemize}
    \item $\P$ ($\NP$) if it is decided by a (non-)deterministic Turing machine in polynomial time,
    \item $\PSPACE$ ($\NPSPACE$) if it is decided by a (non-)deterministic Turing machine in polynomial space,
    \item $\APSPACE$ if it is decided by an alternating Turing machine in polynomial space,
    \item $\EXP$ if it is decided by a deterministic Turing machine in time $2^{p(n)}$ for a polynomial $p$.
\end{itemize}

Alternating Turing machines are a generalization of non-deterministic machines. They are introduced by Chandra, Kozen, and Stockmeyer \cite{alternation}, who also proved that $\APSPACE = \EXP$.

\medskip

To compare the computational complexity of decision problems, we use the notion of \emph{reductions}. Let $A, B$ be computational problems. If there is a function $r$ computable by a Turing machine in logarithmic space such that $x \in A \Leftrightarrow r(x) \in B$, then we call $r$ a \emph{logspace reduction} from $A$ to $B$. A Turing machine works in logarithmic space if on any input $w$ its tapes are restricted to size $\bigO{\log \size{w}}$, except a read-only input tape and a special output tape where the head cannot move to the left.

We say that $A$ is \emph{logspace-reducible} to $B$, written $A \leqlogm B$, if there exists a logspace reduction from $A$ to $B$.
Problems $A$ and $B$ such that $A \leqlogm B$ and $B \leqlogm A$ are called \emph{logspace-equivalent}.
We say that a problem $A$ is $\leqlogm$-\emph{hard} for a class $\mathcal{C}$ if $B \in \mathcal{C}$ implies $B \leqlogm A$, and $\leqlogm$-\emph{complete} for $\mathcal{C}$ if $A \in \mathcal{C}$ and $A$ is $\leqlogm$-hard for $\mathcal{C}$.
For the sake of brevity, we write simply $\leq$ instead of $\leqlogm$ and just say that a problem is \emph{hard} or \emph{complete}, respectively.

\subsection*{Boolean functions}

We call a \emph{Boolean function} any function of the form $f \colon \{0,1\}^n \to \{0,1\}$, where $\arity{f} \dfn n \in \N$ is the \emph{arity of $f$}. It can be zero; there are exactly two such constant Boolean functions, \emph{truth} $\top$ and \emph{falsity} $\bot$.

A Boolean function $f$ is \emph{monotone in its $i$-th argument}, where $1 \leq i \leq \arity{f}$, if $a_i \leq a_i'$ implies $f(a_1,\ldots,a_i,\ldots,a_n) \leq f(a_1,\ldots,a_i',\ldots,a_n)$.
For example, $0 \leq 1$, but if $f$ is the Boolean implication $\imp$, then it holds $f(0,0) \not\leq f(1,0)$, so $\imp$ is not monotone in its first argument.
A function that is monotone in all arguments is \emph{monotone}.

\smallskip

A full classification of all Boolean functions was accomplished by Post \cite{Post41} with the concept of \emph{clones}. A clone $C$ is a set of Boolean functions that is closed under composition and projection to arguments. The smallest clone containing a set of Boolean functions $B$ is written $[B]$, and $B$ is then called a \emph{base} (of $[B]$). Post proved that every clone has a finite base, and for this reason we use only \emph{finite} sets  as bases.

In this work we focus on the clones $\CloneBF \dfn [\Set{\land, \neg}]$ and $\CloneSE \dfn [\Set{\nimp}]$. $\CloneBF$ is the largest clone in Post's lattice, as all Boolean functions can be built from $\{\land,\neg\}$; $\CloneBF$ is also called \emph{expressively complete}. $\CloneSE$ is the clone of all so-called \emph{1-separating functions}. We prove, analogously to Lewis's result in propositional logic, that these clones induce equal lower bounds regarding the computational complexity of the corresponding (temporal) satisfiability problem.

\subsection*{Computation Tree Logic and its syntactical fragments}

Computation Tree Logic (CTL) extends classical modal logic; as atoms we use a countable infinite set of \emph{atomic propositional statements} $\PS \dfn \{p_1, p_2, \ldots\}$, denoted by Latin letters.

Given a base $C$, $\calB(C)$ denotes the corresponding fragment of CTL using only the Boolean connectives in $C$.
With this notation we follow Allen Emerson, Halpern and Schnoebelen \cite{Emerson86,schnoeb} with "$\calB$" for \emph{branching time}, but generalize the notation to accommodate different bases.
The set of \emph{CTL formulas over $C$} is generated by the following grammar:
\begin{align*}
\varphi &\ddfn p \mid f(\underbrace{\varphi,\ldots,\varphi}_{\mathclap{\arity{f}\text{ many}}}) \mid \A\psi \mid \E\psi\\
\psi &\ddfn \X \varphi \mid \F \varphi \mid \G \varphi \mid [\varphi \U \varphi] \mid [\varphi \RLS \varphi]\text{,}
\end{align*}
where $p \in \PS$, $f \in C$. The symbols $\A$ and $\E$ are called \emph{path quantifiers}, and in CTL formulas, they are always followed by $\X, \F, \G, \U$ or $\RLS$, which are called \emph{temporal operators}. The \emph{set of CTL operators} is
\[
\TL \dfn \Set{ QO | Q \in \Set{\A,\E}, O \in \Set{ \X, \F, \G, \U, \RLS}}\text{.}
\]

Note that the binary operators $\U$ and $\RLS$ are used in infix notation: we write \eg $\A[\varphi\U\psi]$ instead of $\A\U(\varphi,\psi)$.
The \emph{duals} of temporal operators resp.\ path quantifiers are $\overline{\A} \dfn \E$, $\overline{\E} \dfn
\A$, $\overline{\F} \dfn \G$, $\overline{\G} \dfn \F$, $\overline{\U} \dfn \RLS$, $\overline{\RLS} \dfn \U$ and $\overline{\X} \dfn \X$.

\smallskip

If $T \subseteq \TL$, then $\calB(C, T)$ is the set of all CTL formulas over $C$, restricted to the CTL operators in $T$ and their duals. We always assume $C$ and $T$ disjoint.

An important property of formulas is their \emph{temporal depth}, which is the maximal nesting depth of temporal operators. It is inductively defined as
\begin{alignat*}{2}
  &\td(p) \; & &\dfn 0 \text{ for } p \in \PS\text{,}\\
  &\td(f(\varphi_1,\ldots,\varphi_{\arity{f}})) \; & &\dfn\max\{0, \td(\varphi_1), \ldots,\td(\varphi_{\arity{f}})\}\text{ for }f \in C\text{,}\\
  &\td(Q\varphi) \; & & \dfn  \td(\varphi)\text{ for }Q \in \{\A, \E\}\text{,}\\
  &\td(O \varphi) \; & &\dfn \td(\varphi) + 1\text{ for }O \in \{\X,\F,\G\}\text{, and}\\
  &\td([\varphi_1 O \varphi_2]) \;&&\dfn \max\{\td(\varphi_1), \td(\varphi_2)\} + 1\text{ for }O \in \{\U,\RLS\}\text{.}
\end{alignat*}

The fragment of $\calB(C, T)$ that contains only formulas of temporal depth at most $i$ is written $\calB_i(C, T)$. We will often omit $T$ if $T = \TL$, and similarly $C$
if $C = \{\land, \lor, \neg\}$. If the meaning is clear, then we omit the curly brackets of the sets $C$ and $T$.

For common Boolean operators, like $\varphi \land \psi$, we use the infix notation. Moreover, we will use abbreviations like $\varphi\rightarrow\psi$ and $\varphi \leftrightarrow \psi$.
Unary operators ($\neg,\X,\F,\G$) take precedence before binary operators,
$\land$ before $\lor$, and $\land,\lor$ before $\rightarrow$ and $\leftrightarrow$.

The set of subformulas of a given formula $\varphi \in \calB(C, T)$ is denoted $\SF{\varphi}$. It is inductively defined as
\begin{alignat*}{2}
  &\SF{p} \; & &\dfn \{\,p\,\} \text{ for } p \in \PS\text{,}\\
  &\SF{f(\varphi_1,\ldots,\varphi_{\arity{f}})} \; & &\dfn\{f(\varphi_1,\ldots,\varphi_{\arity{f}})\} \cup \bigcup_{\mathclap{1 \leq i \leq n}} \SF{\varphi_i}\text{ for }f \in C\text{,}\\
  &\SF{QO \varphi} \; & & \dfn \{QO\varphi\} \cup \SF{\varphi}\text{ for unary }QO \in T\text{,}\\
  &\SF{Q[\varphi_1 O \varphi_2]} \;&&\dfn \{Q[\varphi_1 O \varphi_2]\} \cup \SF{\varphi_1} \cup \SF{\varphi_2}\text{ for binary }QO \in T\text{.}
\end{alignat*}

\subsection*{Kripke structures}

A \emph{Kripke frame} is a directed graph $(W, R)$, where $W$ is
the set of \emph{worlds} or \emph{states}, and $R \subseteq W \times W$ is the
\emph{successor} relation. The reflexive, transitive closure of $R$ is denoted $R^*$. We say that $u$ is \emph{reachable} from $v$ if $vR^*u$.

A \emph{Kripke structure} is a tuple $\calK = (W, R, V)$ where $(W,R)$ is a Kripke frame, and $V \colon \PS \to \powerset{W}$ is its \emph{valuation function} that
maps to each atomic proposition a subset of worlds. Intuitively, the proposition $p$ "holds" in the worlds $w \in V(p)$. The set $\Set{p \in \PS | w \in V(p)}$ of propositions holding in a world $w$ is sometimes also called the \emph{labeling} of $w$ in $\calK$, and if a proposition $p$ is in this set then we say that $p$ is \emph{labeled in $w$}.
Finally, a \emph{rooted Kripke structure} is a tuple $\calM = (W, R, V, w)$ where $(W,R,V)$ is a Kripke structure and $w \in W$ is called the \emph{root} of $\calM$.

For the semantics of CTL, we consider infinite paths through the underlying Kripke frame of a structure. Given a Kripke frame $F = (W,R)$, a \emph{path} $\pi$ through $F$ is an infinite sequence $\pi = (w_0, w_1, w_2, \ldots)$ of worlds $w_i \in W$ such that $w_iRw_{i+1}$ for all $i \geq 0$. Paths through (rooted) Kripke structures are defined accordingly.

Define $\pi[i] \dfn w_i$ as the $i$-th world of $\pi$, where $\pi[0]$ is the \emph{origin} of $\pi$, and $\pi_{\geq k} \dfn (\pi[k], \pi[k+1], \ldots)$ for all $k \geq 0$ are the \emph{suffixes} of $\pi$. Conversely, any finite sequence $(\pi[0], \pi[1], \ldots, \pi[k])$ is a \emph{prefix} of $\pi$.
Moreover, if $0 \leq i_1 < i_2 < \ldots$, then $(\pi[i_1], \pi[i_2], \ldots)$ is a \emph{subpath} of $\pi$.

The set of all paths through $\calK$ with origin $w$ is written $\Pi^\calK(w)$, or just $\Pi(w)$ if $\calK$ is clear.
A Kripke frame resp.\ (rooted) structure is \emph{serial} if every $w \in W$ has at least one $R$-successor.
A rooted Kripke structure $(W,R,V,w)$ is \emph{$R$-generable} if every $w' \in W$ is reachable from $w$.

\medskip

The semantics of CTL on Kripke structures can now be defined inductively. Here, $\calK = (W,R,V)$ is a serial Kripke structure, $w \in W$, $\pi$ is a path through $\calK$, $f$ is an $n$-ary Boolean function and $\vec{b} = (b_1,\ldots,b_n) \in \{0,1\}^n$ is a Boolean vector:
\begin{alignat*}{2}
  &(\calK,w) \vDash p \quad \text{for }p \in \PS \quad & &\text{ iff } w
  \in V(p)\\
  &(\calK,w) \vDash f(\varphi_1, \ldots, \varphi_n) & &\text{
  iff }\exists \vec{b} : f(\vec{b}) = 1 \text{ and } \forall i \in
  [n] :
  b_i = 1 \Leftrightarrow (\calK,w) \vDash \varphi_i\\
  &(\calK,w) \vDash \A \psi & &\text{ iff }
  \forall \pi \in \Pi(w) \; \colon \; (\calK, \pi) \vDash \psi\\[1mm]
  &(\calK,\pi) \vDash p \quad \text{for }p \in \PS \quad &
  &\text{ iff } \pi[0] \in V(p)\\
  &(\calK,\pi) \vDash f(\varphi_1, \ldots, \varphi_n) & &\text{
  iff }\exists \vec{b} : f(\vec{b}) = 1 \text{ and } \forall i \in
  [n] :
  b_i = 1 \Leftrightarrow (\calK,\pi) \vDash \varphi_i\\
  &(\calK,\pi) \vDash \A \psi & &\text{ iff }
  (\calK,\pi[0]) \vDash \A \psi\\
  &(\calK,\pi) \vDash \X \psi & &\text{ iff }
  (\calK, \pi_{\geq 1}) \vDash \psi\\
  &(\calK, \pi) \vDash \psi \U \psi' & &\text{ iff }
  \exists i \geq 0 \; \colon  (\calK, \pi_{\geq i}) \vDash \psi'\text{ and } \forall j < i \; \colon (\calK,
  \pi_{\geq j}) \vDash \psi\end{alignat*}
The remaining operators are treated follows:
Interpret $\E\psi$ as $\neg \A \neg\psi$, $\G\psi$ as $\neg\F\neg \psi$, $\F\psi$ as $\top\U \psi$, and $\varphi\RLS\psi$ as $\neg[\neg\varphi \U \neg\psi]$.
If the Kripke structure $\calK$ is clear from the context, we simply write $w \vDash \varphi$
or $\pi \vDash \varphi$ instead of $(\calK,w) \vDash \varphi$ and
$(\calK,\pi) \vDash \varphi$.

If $\varphi$ and $\psi$ are CTL formula, then $\varphi$ \emph{implies} or \emph{entails} $\psi$, written $\varphi \vDash \psi$, if $\calM \vDash \varphi$ implies $\calM \vDash \Psi$ for all rooted serial Kripke structures $\calM$.
$\varphi$ and $\psi$ are \emph{equivalent}, written $\varphi \equiv \psi$, if $\varphi \vDash \psi$ and $\psi \vDash \varphi$.

If the necessary Boolean functions are available, many sets of CTL operators can be defined by smaller sets. For instance, formulas using the operator set $\{\AF, \AU, \EG, \E\RLS \}$ can be rewritten to use only $\{\AU \}$
when the connectives $\neg$ and $\lor$ are allowed. For this reason, we will denote all CTL fragments by stating a defining set of universally quantifying CTL operators, like $\{\AU\}$.

If $\Phi \subseteq \B$ is a CTL fragment, then $\SAT(\Phi)$ is the set of all \emph{satisfiable} formulas $\varphi \in \Phi$, \ie, for which there is a rooted serial Kripke structure $\calM$  such that $\calM \vDash \varphi$. Call any such structure a \emph{model} of $\varphi$.
Obviously every serial rooted Kripke structure contains a serial, $R$-generable rooted structure that satisfies the same set of CTL formulas.
 
\section{Complexity of CTL and its temporal operator fragments}\label{sec:ctl-fragments}

To measure the complexity of a fragment of CTL, we require a sensible notion of the \emph{length} of a formula. We define the length $\size{\varphi}$ of $\varphi$ as the number of symbols in $\varphi$, where any CTL operator, Boolean connective, proposition and parenthesis is counting as one symbol.

\smallskip

In the following, we introduce the idea of \emph{optimal model size} and \emph{optimal model extent}. For the different fragments of CTL, these measures range between constant and exponential, and also influence the computational complexity of the corresponding satisfiability problem.

\begin{definition}[Size and extent]\label{def:size-diam}
Let a Kripke frame $F = (W, R)$ be $R$-generable and non-empty.
The \emph{size} of $F$ is the number $\size{W}$ of worlds.
The \emph{extent} of $F$ is the greatest $n \in \N$ such that some $R$-path visits $n + 1$ distinct vertices.

The size and extent of a (rooted) Kripke structure is defined as the size and extent of the underlying frame.
\end{definition}

For instance, in a finite directed tree, the extent equals its depth. The difference to, say, the diameter of a graph\footnote{The maximal length of a shortest path between two vertices} is that transitive edges reduce the diameter, but not the extent. This distinction is important, as several CTL operators cannot differentiate between a structure and its transitive closure. For this reason, the diameter of models cannot be a meaningful measure in the classification of CTL fragments.

\begin{definition}[Optimal model size and extent]
Let $\Phi \subseteq \B$ be a set of satisfiable CTL formulas. Let $\sigma \colon \N \to \N$.

\begin{itemize}
    \item $\sigma$ is a \emph{model size upper bound} of $\Phi$ if every satisfiable $\varphi \in \Phi$ has a model of size at most $\bigO{\sigma(\size{\varphi})}$.
    \item $\sigma$ is a \emph{model size lower bound} of $\Phi$ if $\Phi$ contains an infinite family of satisfiable formulas $\varphi_1,\varphi_2,\ldots$ such that each $\varphi_i$ has only models of size at least $\bigOmega{\sigma(\size{\varphi_i})}$.
    \item $\sigma$ is an \emph{optimal model size} of $\Phi$ if it is both an upper and lower bound.
\end{itemize}

Similarly define \emph{model extent upper/lower bound} and \emph{optimal model extent} $\epsilon$.
\end{definition}
As no path can visit more distinct vertices than the frame contains, it follows that size forms an upper bound for extent.

\smallskip

An exponential model size upper bound for full CTL was proven by Allen Emerson and Halpern \citeref{halpern}{Thm. 4.1.}. Although they did not consider Boolean clones, the proof indeed works independently of the particular clone.

\begin{theorem}[Small model property of CTL
 \cite{emersonTemporal}]\label{thm:upper-size-all}
$\B(C,T)$ has optimal model size of at most $2^\bigO{n}$ for every base $C$ and $T \subseteq \TL$.
\end{theorem}

A deterministic exponential time algorithm for satisfiability of \emph{propositional dynamic logic (PDL)}, which subsumes CTL, was given by Pratt \cite{Pratt1980231}. Allen Emerson and Halpern presented a similar algorithm for CTL directly; it constructs a structure of exponential size to check the satisfiability of the formula \citeref{halpern}{Thm. 5.1.}. See also Allen Emerson \cite{emersonTemporal}.

\begin{theorem}[\cite{emersonTemporal,halpern,Pratt1980231}]\label{thm:all-c-in-exp}
$\SAT(\B) \in \EXP$.
\end{theorem}

In this rest of this section, for every temporal operator
fragment of CTL, these upper bounds of the computational complexity are either improved, or proven tight.
We begin by showing the lower bound for the PSPACE-complete fragment $\B(\AF)$.

\subsection{The AF fragment}\label{sec:af}

For the hardness of $\SAT(\B(\AF))$, we consider a reduction from the PSPACE-complete problem of \emph{quantified Boolean formulas} (qbfs).
The grammar of qbfs is
\[
\varphi \ddfn \varphi \land \varphi \mid \neg \varphi \mid \forall p\, \varphi \mid \exists p\, \varphi \mid p\text{,}
\]
where $p \in \PS$. The semantics are defined via \emph{Boolean assignments}, which are functions $\theta \colon \Phi \to \{0,1\}$ for finite $\Phi \subseteq \PS$. In particular, for a Boolean assignment $\theta$ it holds $\theta \vDash \forall p \,\varphi$ if $\theta^p_b \vDash \varphi$ for all $b \in \{0,1\}$, where $\theta^p_b(p) \dfn b$ and $\theta^p_b(q) \dfn \theta(q)$ for $p \neq q$. $\exists p \,\varphi$ behaves like $\neg \forall p \neg \varphi$, and the other connectives are defined as in propositional logic. Say that a qbf $\varphi$ is \emph{closed} if has no free variables, and say that a qbf is \emph{true} if it is closed and satisfied by some Boolean assignment.

The corresponding computational problem is:
\[
\TQBF \dfn \Set{ Q_1 x_1 \ldots Q_n x_n \psi | \begin{array}{l}
\{ Q_1, \ldots, Q_n \} \subseteq\{\exists, \forall\}, \{x_1,\ldots,x_n\}\subseteq \PS, \psi \in \B_0\\
\text{and }Q_1 x_1 \ldots Q_n x_n \psi \text{ is a closed, true qbf}
\end{array}}
\]

\begin{theorem}[Meyer and Stockmeyer \cite{stockmeyer_word_1973}]
    $\TQBF$ is $\PSPACE$-complete.
\end{theorem}

A formula in the above form, with all quantifiers at the beginning, is called \emph{prenex form}, with \emph{prefix} $Q_1 x_1 \ldots Q_n x_n$ and \emph{matrix} $\psi \in \B_0$.

\smallskip

The following is an alternative definition of the truth of qbfs; it is helpful in the subsequent reduction to CTL.

\begin{definition}\label{def:witness-seq}
Let $\varphi = Q_1 x_1 \ldots Q_n x_n \psi$ be a closed qbf.
A \emph{proof tree} $T = (\Theta, E)$ for $\varphi$ is a tree of Boolean assignments that meets the following conditions:
\begin{enumerate}
    \item the everywhere undefined assignment $\theta_0 \in \Theta$ is the root of $T$,
    \item if $\theta : \{x_1, \ldots, x_{m-1}\} \to \{0,1\}$ is in $\Theta$, $m \leq n$, and $Q_m = \forall$ \emph{($\exists$)}, then for all (some) $b \in \{0,1\}$, $\theta^{x_{m+1}}_b \in \Theta$ and $(\theta,\theta^{x_{m+1}}_b)\in E$,
    \item if $\theta : \{x_1,\ldots,x_n\} \to \{0,1\}$ is in $\Theta$, then $\theta \vDash \psi$.
\end{enumerate}
\end{definition}
Intuitively, (1) describes the empty Boolean assignment, (2) simulates universal and existential branching with respect to the Boolean quantifiers, and (3) states that the matrix is true under the "leaf" Boolean assignments $\theta$. It is straightforward to show by induction:

\begin{proposition}
A closed qbf is true if and only if it has a proof tree.
\end{proposition}

It follows the hardness proof of $\SAT(\B_2(\AF))$ by reduction from TQBF.
The standard reduction from TQBF to modal satisfiability (see Ladner~\cite{ladner77}) would be to span a proof tree of exponential size, with the help of modal operators, directly in a Kripke structure.
This approach, however, does not work here as the operators $\AF$ and $\EG$ have ``mixed'' path and state quantifiers, \ie, whenever the path quantifier is universal, then the state on this path is quantified existentially, and vice versa. This leaves no sensible way to span a tree of exponential size.
Consequently, the proof tree has to be encoded on a single path in a complicated manner.

\begin{theorem}\label{thm:qbf-to-afeg}
$\SAT(\B_2(\AF))$ is $\PSPACE$-hard.
\end{theorem}
\begin{proof}
Let $\varphi = Q_1 x_1 \ldots Q_n x_n \psi$ be a closed qbf.
The reduction maps $\varphi$ to a formula $\varphi^* \in \B_2(\AF)$ that is satisfiable if and only if $\varphi$ is true.

The idea is to enforce a "flattened" proof tree as a long path inside the model. The given path is successively subdivided into segments: the first half should uniformly set $x_1$ true, while on the other half $\neg x_1$ holds. Each of the segments is then again divided to account for the possible truth values of $x_2$, and so on. Figure~\ref{fig:af-idea} illustrates this construction.

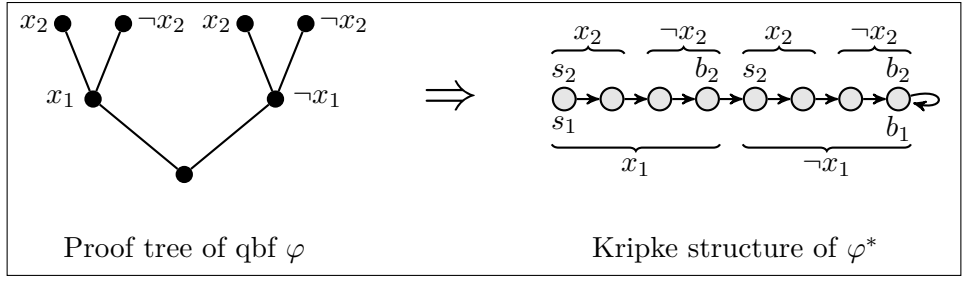
\begin{figure}\centering
\begin{tikzpicture}[framed,->,>=stealth',shorten >=0pt,auto,node distance=1.3cm,        thick,         world/.style={circle,fill=black!10,draw,minimum size=0.3cm,inner sep=0pt},         world2/.style={circle,dotted,draw,minimum size=0.3cm,inner sep=0pt},         tree/.style={circle,fill=black,draw,minimum size=0.2cm,inner sep=0pt},         dummy/.style={inner sep=0pt,node distance=.2cm}]

\node[tree] (t0) at (0,0) {};
\node[tree] (t1) at (-1.2cm,1cm) {};
\node[tree] (t2) at (1.2cm,1cm) {};
\node[tree] (t3) at (-1.6cm,2cm) {};
\node[tree] (t4) at (-0.8cm,2cm) {};
\node[tree] (t5) at (0.8cm,2cm) {};
\node[tree] (t6) at (1.6cm,2cm) {};

\node[dummy] (x1) [left = 1mm of t1] {$x_1$};
\node[dummy] (nx1) [right = 1mm of t2] {$\neg x_1$};
\node[dummy] (x1x2) [left = 0.5mm of t3] {$x_2$};
\node[dummy] (nx1x2) [left = 0.5mm of t5] {$x_2$};
\node[dummy] (x1nx2) [right = 0.5mm of t4] {$\neg x_2$};
\node[dummy] (nx1nx2) [right = 0.5mm of t6] {$\neg x_2$};

\path[-]
(t0) edge (t1)
(t0) edge (t2)
(t1) edge (t3)
(t1) edge (t4)
(t2) edge (t5)
(t2) edge (t6)
;

\node[world] (w1) at (5cm,1cm) {};
\node[world] (w2) [right = 3mm of w1] {};
\node[world] (w3) [right = 3mm of w2] {};
\node[world] (w4) [right = 3mm of w3] {};
\node[world] (w5) [right = 3mm of w4] {};
\node[world] (w6) [right = 3mm of w5] {};
\node[world] (w7) [right = 3mm of w6] {};
\node[world] (w8) [right = 3mm of w7] {};

\node[dummy] (l1) [below = 0.5mm of w1] {$s_1$};
\node[dummy] (l2) [below = 0.5mm of w8] {$b_1$};
\node[dummy] (l3) [above = 0.5mm of w1] {$s_2$};
\node[dummy] (l4) [above = 0.5mm of w5] {$s_2$};
\node[dummy] (l5) [above = 0.5mm of w4] {$b_2$};
\node[dummy] (l6) [above = 0.5mm of w8] {$b_2$};

\path[->]
(w1) edge (w2)
(w2) edge (w3)
(w3) edge (w4)
(w4) edge (w5)
(w5) edge (w6)
(w6) edge (w7)
(w7) edge (w8)
(w8) edge [loop right] (w8)
;

\draw[-] [
    thick,
    decoration={
        brace,
        mirror,
        raise=0.6cm
    },
    decorate
] (w1.west) -- (w4.east)
node [pos=0.5,anchor=north,yshift=-0.65cm] {$x_1$};

\draw[-] [
    thick,
    decoration={
        brace,
        mirror,
        raise=0.6cm
    },
    decorate
] (w5.west) -- (w8.east)
node [pos=0.5,anchor=north,yshift=-0.65cm] {$\neg x_1$};

\draw[-] [
    thick,
    decoration={
        brace,
        raise=0.6cm
    },
    decorate
] (w1.west) -- (w2.east)
node [pos=0.5,anchor=north,yshift=1.1cm] {$x_2$};

\draw[-] [
    thick,
    decoration={
        brace,
        raise=0.6cm
    },
    decorate
] (w3.west) -- (w4.east)
node [pos=0.5,anchor=north,yshift=1.1cm] {$\neg x_2$};

\draw[-] [
    thick,
    decoration={
        brace,
        raise=0.6cm
    },
    decorate
] (w5.west) -- (w6.east)
node [pos=0.5,anchor=north,yshift=1.1cm] {$x_2$};

\draw[-] [
    thick,
    decoration={
        brace,
        raise=0.6cm
    },
    decorate
] (w7.west) -- (w8.east)
node [pos=0.5,anchor=north,yshift=1.1cm] {$\neg x_2$};

\node[dummy] (left) at (0,-1cm) {Proof tree of qbf $\varphi$};

\node[dummy] (right) at (7.25cm,-1cm) {Kripke structure of $\varphi^*$};

\node[dummy] (arrow) at (3.5cm, 1cm) {\huge$\Rightarrow$};

\end{tikzpicture}\caption{\label{fig:af-idea}Sketch of the reduction from $\TQBF$ to $\SAT(\B(\AF))$}
\end{figure}

The implementation uses several auxiliary propositions. The variables $t_i$ and $t'_i$ span an interval on the path where $x_i$ is true. Conversely, $f_i$ and $f'_i$ span an interval where $x_i$ is false. The actual truth resp.\ falsity of $x_i$ in these segments is enforced by the formula $\gamma_i$.
The formulas $\alpha_i^\forall$ and $\alpha_i^\exists$ are responsible for the mentioned subdivision of a path: in one case, both the "true" and "false"
subsegments are forced to appear in this order. In the second case, one can be chosen.

Intuitively, the propositions $s_i$ and $b_i$ have the following meaning. Every occurrence of $s_i$ on the path \emph{starts} the subdivision into either one or two subsegments with respect to the truth of $x_i$. Any occurrence of $b_i$ \emph{blocks} all imposed $\AF$s containing $\neg b_i$, such as in the $\alpha$-subformulas. This is due to $b_i \rightarrow \EG b_i$ holding everywhere on the path, and the fact that the $\AF$s in $\alpha_i^{Q_i}$ are of the form $\AF(\ldots \sland \neg b_i)$.

As a result, the formula $\beta_i$ ensures that there is no "overlapping" of segments: the $\AF$-subformulas of $\alpha^{Q_i}_i$ are fulfilled on the path exactly in the order as they appear in the formula. Furthermore, the subdivisions for $x_{i+1}$ between $t_i$ and $t'_i$ resp.\ $f_i$ and $f'_i$ are contained inside these segments.

The proposition $e$ simply enforces the initial $b_1$ to appear on the path.
The complete formula $\varphi^*$ is defined as

\begin{align*}
\varphi^* &\dfn s_1 \sland\AF e \sland \EG \bigg[\psi\, \sland \,(e \imp b_1)  \, \sland \,\bigand_{i=1}^n
\big(\alpha_i  \, \sland \,\beta_i \,\sland
\,\gamma_i \big) \bigg]\text{,}
\end{align*}
where $\alpha_i \dfn \alpha^{Q_i}_i$,
\begin{alignat*}{4}
&\alpha^\forall_i \dfn \;&&
\big(s_i \rightarrow \AF (t_i \sland \neg b_i)\big) \;\sland \; && \qquad\qquad  \alpha^\exists_i \dfn \;  &&\big(s_i \rightarrow \AF((t_i \lor f_i)\sland \neg b_i)\big) \;\sland \;\\
& && \big(t_i \rightarrow (s_{i+1} \sland \AF(t'_i  \sland \neg b_i))\big)\;\sland \;
&& && \big(t_i \rightarrow (s_{i+1} \sland \AF (t'_i \sland \neg b_i))\big)\;\sland \;
\\
& &&\big(t'_i \rightarrow (b_{i+1} \sland \AF (f_i \sland \neg b_i))\big)\;\sland \;
&& && \big(t'_i \rightarrow b_{i+1}\big)\;\sland \;\\
& && \big(f_i \rightarrow (s_{i+1} \sland \AF (f'_i \sland \neg b_i))\big)\;\sland \;
&& &&\big(f_i \rightarrow (s_{i+1} \sland \AF (f'_i \sland \neg b_i))\big)\;\sland \;\\
& &&\big(f'_i \rightarrow b_{i+1}\big)
&& && \big(f'_i \rightarrow b_{i+1}\big)
\end{alignat*}
and
\begin{align*}
\beta_i \dfn &\;(b_i \rightarrow \EG b_i) \\
\gamma_i \dfn &\;\big(\AF t'_i \rightarrow x_i\big) \sland
\big(\AF f'_i \rightarrow \neg x_i\big)\text{.}
\end{align*}

It is easy to show that $\alpha_i$, $\beta_i$ and $\gamma_i$ are all logspace-constructible.
The following lemmas prove the correctness of the reduction.\end{proof}
First we prove that there are in fact the required intervals with $x_i$ being true resp. false between occurrences of $s_i$ and $b_i$.
Let $\calM$ be a model of $\varphi^*$ and $\pi$ a path through it.
Say that $x_i \in \PS$ is \emph{uniformly true} (resp.\ \emph{uniformly false}) on a sequence $\rho = (\pi[j], \ldots, \pi[k])$ of worlds if $\pi[o] \vDash x_i$ (resp.\ $\pi[o]\vDash \neg x_i$) for all $o \in \{j,\ldots, k\}$.

For sequences $\rho = (\pi[j],\ldots,\pi[k])$ and $\rho' = (\pi[j'],\ldots,\pi[k'])$, $\rho$ \emph{contains} $\rho'$ if $j \leq j' \leq k' \leq k$.
For $j \leq k$, call a subsequence $\rho = (\pi[j],\ldots, \pi[k])$ an \emph{$m$-segment} of $\pi$ if $\pi[j]\vDash s_m$, $\pi[k]\vDash b_m$, and $x_i$ is uniformly true or uniformly false on $\rho$ for all $i \in [m-1]$.

\begin{lemma}\label{lem:subdivisions}
Let $\calM$ be a model of $\varphi^*$ and $\pi$ a path through it that satisfies the outermost $\EG$ operator.
Let $\rho$ be an $m$-segment on $\pi$.

If $Q_m = \exists$, then $\rho$ contains an $(m+1)$-segment.
If $Q_m = \forall$, then $\rho$ contains an $(m+1)$-segment where $x_m$ is uniformly true and another $(m+1)$-segment where $x_m$ is uniformly false.
\end{lemma}
\begin{proof}
Let $\rho = (\pi[j], \ldots, \pi[k])$ be an $m$-segment.
Then $\pi[j] \vDash s_m$ and $\pi[k] \vDash b_m$.
For the rest of the proof, suppose $Q_m = \forall$ (the case $Q_m = \exists$
is handled similarly).

Let $o_1 \geq j$ be the smallest number such that $\pi[o_1]\vDash t_m \land \neg b_m$, and similarly $o_2 \geq o_1$ the smallest such that $\pi[o_2]\vDash t'_m \land \neg b_m$; $o_3 \geq o_2$ such that $\pi[o_3] \vDash f_m\land \neg b_m$; and $o_4 \geq o_3$ such that $\pi[o_4]\vDash f'_m\land \neg b_m$.
These worlds occur on $\pi$ in this order due to $\pi[j]\vDash \alpha_m^\forall$.

Next, we prove $o_1,\ldots,o_4 \leq k$. Assume for the sake of contradiction that, \eg, $o_3 \leq k < o_4$. ($o_2 \leq k < o_3$ etc.\ lead to a contradiction analogously.)
For all $o_3 \leq o < o_4$, it holds $\pi[o] \vDash (\neg f'_m \lor b_m)$ by definition of $o_4$.
Furthermore, $\pi[k] \vDash \EG b_m$ by $\beta_m$. Consequently, $\pi[o_3]\vDash \EG b_m$,
contradicting $\pi[o_3] \vDash \AF(f'_m \land \neg b_m)$.

These subsegments of $\rho$ have correct "delimiters" due to $\alpha^\forall_m$: the first one, as $\pi[o_1] \vDash s_{m+1}$ and $\pi[o_2] \vDash b_{m+1}$, and the second one, as $\pi[o_3]\vDash s_{m+1}$ and $\pi[o_4]\vDash b_{m+1}$.
In order to prove that $(\pi[o_1],\ldots,\pi[o_2])$ and $(\pi[o_3],\ldots,\pi[o_4])$ are the desired $(m+1)$-segments, by $\gamma_m$ it suffices to show that $\pi[o] \vDash \AF t'_m$ for all $o_1 \leq o \leq o_2$. (Showing $\pi[o] \vDash \AF f'_m$ for all $o_3 \leq o \leq o_4$ again works similarly.)

For the sake of contradiction, suppose there exists $o \in \{ o_1, \ldots, o_2\}$ such that $\pi[o] \vDash \EG \neg t'_m$. Clearly $o \neq o_2$, as $\pi[o_2] \vDash t'_m$. But then all worlds between $\pi[o_1]$ and $\pi[o]$ satisfy $(\neg t_m' \lor b_m)$, so $\pi[o_1]\vDash \EG (\neg t_m' \lor b_m)$, contradiction.
\end{proof}

In what follows, we say that a world $w$ \emph{agrees} with some assignment $\theta \colon \{x_1, \ldots, x_m\} \to \{0,1\}$, in symbols $w \vdash \theta$, if $\theta(x_i) = 1 \Leftrightarrow w \vDash x_i$ for all $i \in \{1,\ldots,m\}$.
Similarly, given an $m$-segment $\rho = (\pi[j], \ldots, \pi[k])$, we say that $\rho$ agrees with $\theta$, in symbols $\rho \vdash \theta$, if the worlds $\pi[j], \ldots, \pi[k]$ all agree with $\theta$.

\begin{lemma}\label{lem:proof-tree-in-af}
If $\varphi^*$ is satisfiable, then $\varphi = Q_1 x_1 \ldots Q_n x_n \psi$ has a proof tree.
\end{lemma}
\begin{proof}
Let $\calM$ be a model of $\varphi^*$, and let $\pi$ be a path through $\calM$ that witnesses the outermost $\EG$ operator in $\varphi^*$.
The following graph $T = (\Theta, E)$ contains a proof tree for $\varphi$. $\Theta$ is the set of all assignments $\theta \colon \{x_1, \ldots, x_{m-1}\} \to \{0,1\}$ for which there is an agreeing $m$-segment $\rho$ on $\pi$, formally
\[
\Theta \dfn \Set{\theta \colon \{x_1, \ldots, x_{m-1}\} \to \{0,1\} | \begin{array}{l}
1\leq m \leq n+1,\; \exists \; m\text{-segment } \rho\\
\text{on }\pi \text{ such that }\rho \vdash \theta\end{array}}\text{.}
\]

The edges are
\[
E \dfn \Set{ (\theta, \theta') \in \Theta^2 | \begin{array}{l}\text{if }\theta \colon \{x_1,\ldots,x_{m-1}\} \to \{0,1\},\\
\text{then }\exists\, b \in \{0,1\} \text{ such that }\theta' = \theta^{x_m}_b\\
\exists\, m\text{-segment} \;\rho,\; \exists (m+1)\text{-segment} \;\rho'\;\\
\text{such that }\rho\text{ contains }\rho', \rho \vdash \theta \text{ and }\rho' \vdash\theta'\end{array} }\text{.}
\]

Following Definition~\ref{def:witness-seq}, we show that $T$ indeed contains a proof tree of $\varphi$.\footnote{$(\Theta,E)$ may have "wrong", successor-free vertices, so it may not be a proof tree. Nevertheless, we can just crop all worlds unreachable from the root to obtain one.} The empty assignment is in $\Theta$, since there is an $1$-segment (with arbitrary assignment) between the root of $\calM$ (which satisfies $s_1$) and the first point of $\pi$ that satisfies $b_1$.

    If $\theta \in \Theta$ for $\theta \colon \{x_1,\ldots,x_{m-1}\} \to \{0,1\}$, then by definition of $\Theta$ there is an $m$-segment $\rho$ with $\rho \vdash \theta$. Assume $Q_m = \exists$. By Lemma~\ref{lem:subdivisions}, $\rho$ contains an $(m+1)$-segment $\rho'$. Then the assignment $\theta' \dfn \theta^{x_m}_0$ (if $x_m$ is uniformly false on $\rho'$) or $\theta' \dfn \theta^{x_m}_1$ (if $x_m$ is uniformly true on $\rho'$) is in $\Theta$, and consequently $(\theta, \theta') \in E$.
The case $Q_m = \forall$ works analogously.

It remains to show that all leaf assignments $\theta : \{x_1, \ldots, x_n\} \to \{0,1\}$ in $\Theta$ actually satisfy the matrix $\psi$. First observe that for each such $\theta$, $\rho\vdash\theta$ for some $(n+1)$-segment $\rho$ of $\pi$. As $\G\psi$ holds on $\pi$, at least one world $\pi[j]$ agrees with $\theta$ on $\{x_1,\ldots,x_n\}$ and still satisfies $\psi$.
\end{proof}

The above lemma shows the first direction of the correctness. For the other direction---$\varphi^*$ actually being satisfiable if $\varphi$ is true---we construct a model of $\varphi^*$ by an inductive approach. For $m = 0,1,\ldots,n$, we define a Kripke structure $\calK_m = (W_m, R_m, V_m)$ such that $\calK_m$ contains an $R_m$-path $\pi_m$ and for all $w \dfn \pi_m[j]$, $j\geq 0$:
\begin{enumerate}
    \item $w \vDash (e \imp b_1) \land \bigwedge_{i=1}^m (\alpha_i \land \gamma_i) \land \bigwedge_{i=1}^{m+1}\beta_i$,
    \item $w$ satisfies exactly one of $\Set{ s_i, b_i | i \in [n]}$,
    \item $w$ satisfies none of $\Set{t_i, t'_i, f_i, f'_i, s_{i+1}, b_{i+1} | i > m}$,
    \item some $\theta \colon \{x_1,\ldots,x_{m}\} \to \{0,1\}$ in $T$ agrees with $w$,    \item $w \vDash (s_1 \land \AF e)$ if $w = \pi_m[0]$.
\end{enumerate}
Since $T$ is a proof tree of $\varphi$, (4) implies $(\calK_n,\pi_n) \vDash \G \psi$. By additionally (1) and (5), $(\calK_n,\pi_n[0])$ is the desired model of $\varphi^*$. The properties (2)--(3) are not directly required, but simplify the inductive step.

\begin{lemma}
Let $T$ be a proof tree. For all $m \in \{0,1,\ldots,n\}$, there is a Kripke structure $\calK_m = (W_m,R_m,V_m)$ satisfying the properties (1)--(5).
\end{lemma}

\begin{proof}Let $\calK_0$ be as in the following picture. It is easy to verify (1)--(5) for $m = 0$, in particular $T$ contains the empty assignment which agrees with all worlds of $\calK_0$.

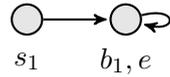
\begin{center}
\begin{tikzpicture}[->,>=stealth',shorten >=1pt,auto,node distance=1.3cm,        thick,         world/.style={circle,fill=black!10,draw,minimum size=0.4cm,inner sep=0pt},         world2/.style={circle,dotted,draw,minimum size=0.5cm,inner sep=0pt},         dummy/.style={inner sep=0pt,node distance=.55cm}]

\node[world] (s0) at (0,0) {};
\node[world] (s1) [right of=s0] {};
\node[dummy] (l0) [below of=s0] {$s_1$};
\node[dummy] (l1) [below of=s1] {$b_1,e$};

\path[->]
(s0) edge node{} (s1)
(s1) edge [loop right] (s1)
;
\end{tikzpicture}
\end{center}

We proceed with the inductive step: assume that $\calK_{m-1} = (W_{m-1},R_{m-1},V_{m-1})$ and an $R_{m-1}$-path $\pi$ exist as above.
First note that, by condition (4) of the induction hypothesis, for all $j$ the tree $T$ contains an assignment $\theta \colon \{x_1, \ldots, x_{m-1}\} \to \{0,1\}$ that agrees with $\pi[j]$. Call that assignment $\theta_{\pi[j]}$ here.

\begin{figure}\centering
\begin{tikzpicture}[framed,->,>=stealth',shorten >=1pt,auto,node distance=1.3cm,        thick,         world/.style={circle,fill=black!10,draw,minimum size=0.4cm,inner sep=0pt},         world2/.style={circle,dotted,draw,minimum size=0.4cm,inner sep=0pt},         dummy/.style={inner sep=0pt,node distance=.6cm}]

\node[world] (s0) at (0,0) {};
\node[world2] (pre) [left of=s0] {};
\node[dummy] (l0) [below of=s0] {$\pi[j]$};
\node[world2] (post) [right of=s0] {};

\node[dummy] (before) at (-2.8cm,0) [text width=2cm] {Before:};

\path[->]
(pre) edge [dotted] (s0)
(s0) edge [dotted] (post)
;

\node[world] (ts0) at (0,-3.5cm) {};
\node[world2] (tpre) [left of=ts0] {};
\node[world] (tu1) [right of=ts0] {};
\node[world] (tu2) [right of=tu1] {};
\node[dummy] (tl0) [below of=ts0] {$\pi[j]$};
\node[dummy] (tul1) [below of=tu1] {$u^j_p$};
\node[dummy] (tul2) [below of=tu2] {$u^j_{p'}$};
\node[world2] (tpost) [right of=tu2] {};
\node[world] (tul) [above=0.6cm of tu2] {};
\node[world] (tug) [above=0.6cm of ts0] {};
\node[dummy] (tll) [above=0.07cm of tul] {$u^j_\ell$};
\node[dummy] (tgl) [above=0.07cm of tug] {$u^j_g$};

\node[dummy] (before) at (-2.8cm,-3.5cm) [text width=3cm] {Case $Q_m = \exists$:\\($p \in\{t,f\}$)};

\path[->]
(tpre) edge [dotted] (ts0)
(ts0) edge node{} (tu1)
(ts0) edge (tug)
(tug) edge [loop right] (tug)
(tu1) edge node{} (tu2)
(tu2) edge [dotted] (tpost)
(tu2) edge (tul)
(tul) edge [loop right] (tul)
;

\node[world] (tfs0) at (0,-7cm) {};
\node[world2] (tfpre) [left of=tfs0] {};
\node[world] (tfug) [above=0.6cm of tfs0] {};
\node[world] (tfu1) [right of=tfs0] {};
\node[world] (tfu2) [right of=tfu1] {};
\node[world] (tfu3) [right of=tfu2] {};
\node[world] (tfu4) [right of=tfu3] {};
\node[dummy] (tfl0) [below of=tfs0] {$\pi[j]$};
\node[dummy] (tful1) [below of=tfu1] {$u^j_t$};
\node[dummy] (tful2) [below of=tfu2] {$u^j_{t'}$};
\node[dummy] (tful3) [below of=tfu3] {$u^j_f$};
\node[dummy] (tful4) [below of=tfu4] {$u^j_{f'}$};
\node[world2] (tfpost) [right of=tfu4] {};

\node[world] (tful) [above=0.6cm of tfu4] {};
\node[dummy] (tfll) [above=0.07cm of tful] {$u^j_\ell$};
\node[dummy] (tfgl) [above=0.07cm of tfug] {$u^j_g$};

\node[dummy] (before) at (-2.8cm,-7cm) [text width=3cm] {Case $Q_m = \forall$:};

\path[->]
(tfpre) edge [dotted] (tfs0)
(tfs0) edge node{} (tfu1)
(tfu1) edge node{} (tfu2)
(tfu2) edge node{} (tfu3)
(tfu3) edge node{} (tfu4)
(tfu4) edge [dotted] (tfpost)
(tfu4) edge (tful)
(tfu2) edge (tful)
(tful) edge [loop right] (tful)
(tfug) edge [loop right] (tfug)
(tfs0) edge (tfug)
;

\end{tikzpicture}\caption{\label{fig:af-construction}Subdivision step from $\calK_{m-1}$ to $\calK_m$}
\end{figure}
 
Define the new structure as follows. Modifications are performed immediately between worlds $\pi[j] \in V_{m-1}(s_m)$ and their successor $\pi[j+1]$. Namely, the edge between them is removed and the substructure depicted in Figure~\ref{fig:af-construction} is inserted. If $Q_m = \exists$ and $\theta \dfn \theta_{\pi[j]}$ has $\theta^{x_m}_1$ as a child in $T$, then insert worlds $u^j_t, u^j_{t'}$. Otherwise insert $u^j_f,u^j_{f'}$, and if $Q_m = \forall$, then insert all four worlds. The worlds $u^j_g, u^j_\ell$ are added  unconditionally. Formally, if $Q_m = \exists$, then
\begin{align*}
    W_m &\dfn W_{m-1} \cup \bigcup\Set{ U^j | \pi[j] \vDash s_m }\\
    R_m &\dfn R_{m-1} \setminus \Set{ (\pi[j], \pi[j+1]) | \pi[j] \vDash s_m } \cup
    \Set{ R^j | \pi[j] \vDash s_m }\text{.}\end{align*}
where $U^j \dfn \{ u^j_p, u^j_{p'}, u^j_g, u^j_\ell \}$ with $p = t$ if $\theta_{\pi[j]}$ has $(\theta_{\pi[j]})^{x_m}_1$ as a child in $T$, and with $p = f$ otherwise, and
\begin{align*}
R^j &\dfn \{ (\pi[j], u^j_p), (\pi[j],u^j_g), (u^j_g,u^j_g), (u^j_p, u^j_{p'}), (u^j_{p'}, u^j_\ell), (u^j_\ell, u^j_\ell), (u^j_{p'}, \pi[j+1])\}\text{.}
\end{align*}

If $Q_m = \forall$, then $U^j \dfn \{ u^j_t, u^j_{t'}, u^j_{f}, u^j_{f'}, u^j_\ell, u^j_g \}$ and
\begin{align*}
R^j &\dfn \left\{ \begin{array}{l}(\pi[j], u^j_t), (u^j_t, u^j_{t'}, (u^j_{t'}, u^j_f), (u^j_{t'}, u^j_{\ell}) (u^j_f, u^j_{f'}),\\
(u^j_{f'}, u^j_{\ell}) (u^j_\ell, u^j_\ell), (u^j_{f'}, \pi[j+1]) \end{array}\right\}\text{.}
\end{align*}

Let $U_z$ denote the set of \emph{all} inserted worlds $u^j_z$, \ie, $U_z \dfn \Set{ u^j_z | \pi[j] \vDash s_m }$ for $z \in \{t,t',f,f',g,\ell\}$.

\medskip

After the worlds and edges, it remains to define the valuation $V_m$:
\begin{alignat*}{4}
&V_m(t_m)&&\dfn U_t \cup U_g \quad \quad && V_m(s_{m+1})&&\dfn U_t \cup U_f\\
&V_m(t'_m) &&\dfn U_{t'} && V_m(b_{m+1}) &&\dfn U_{t'} \cup U_{f'} \cup U_\ell\\
&V_m(f_m) &&\dfn U_f \cup U_\ell && V_m(b_i) && \dfn V_{m-1}(b_i) \text{ for }i \leq m\\
&V_m(f'_m) &&\dfn U_{f'} &&  V_m(e) &&\dfn V_{m-1}(e)\cup U_g \cup U_\ell
\end{alignat*}
The assignments to $x_1, \ldots, x_{m-1}$ are expanded to $x_m$ as follows:
\[
V_m(x_m)\dfn U_t \cup U_{t'} \cup \Set{\pi[j] | \theta_{\pi[j]} \text{ has } (\theta_{\pi[j]})^{x_m}_1 \text{ as child in }T }\text{,}
\]
and for $i < m$, the value of $x_i$ is just "copied" to the inserted worlds:
\[
V_m(x_i)\dfn V_{m-1}(x_i) \cup \Set{ u^j_p | p \in \{t,t',f,f'\}, \pi[j] \in V_{m-1}(x_i)}\text{.}
\]
For all other propositions $p$, let $V_m(p) \dfn V_{m-1}(p) \cup U_g \cup U_\ell$.

\medskip

Define $\pi^*$ as the path through $\calK_m$ that is obtained from $\pi$ by replacing every edge $(\pi[j],\pi[j+1]) \in R_{m-1} \setminus R_m$ with the corresponding sequence of new $R_m$-edges, \eg, $(\pi[j],u^j_t,u^j_{t'},\pi[j+1])$.
It is straightforward to verify the properties (2)--(5) in $\pi^*$.
We proceed by showing property (1), \ie, that $\pi^*[j]$ satisfies $(e \imp b_1)$, $(\bigwedge_{i=1}^m\alpha_i \land \gamma_i)$, and $(\bigwedge_{i=1}^{m+1}\beta_i)$ for all $j\geq 0$. For the proof, we distinguish between \emph{old} worlds $w \in W_{m-1}$ and \emph{new} worlds $w \in U_t \cup U_{t'} \cup U_f \cup U_{f'}$. All worlds on $\pi^*$ are either old or new. Furthermore, it is easy to verify in $V_m$ that all new worlds satisfy $\left(\bigwedge_{i=1}^m\alpha_i \right)$, $\left(\bigwedge_{i=1}^{m+1}\beta_i\right)$ and $\neg e$.

They also satisfy $\gamma_i \equiv (x_i \imp \EG \neg f'_i) \land (\neg x_i \imp \EG \neg t'_i)$, which can be seen as follows. For $i = m$, Figure~\ref{fig:af-construction} shows the path from $u^j_t$ and $u^j_{t'}$ to $u^j_\ell$ satisfying $\G \neg f'_m$, and a similar path from $u^j_f$ and $u^j_{f'}$ satisfying $\G \neg t'_m$, or both if $Q_m = \forall$.
For $i < m$, recall that $u^j_t,u^j_{t'},u^j_f,u^j_{f'}$ agree with $\pi[j]$ on the value of $x_i$. If \eg $x_i = 0$, then by induction hypothesis, $(\calK_{m-1},\pi[j]) \vDash \EG \neg t'_i$ via some path $\pi' = (\pi[j],\pi[j+1], \ldots)$.
Clearly $\pi'$ can be extended to an $R_m$-path $(\pi[j], \ldots, \pi[j+1],\ldots)$ satisfying $\G \neg t'_i$.

\smallskip

With the inductive step on the new worlds being settled, assume for the rest of the proof that $w$ is old. By induction hypothesis (1), then $w \vDash (e \imp b_1)$. Furthermore, $w \vDash \alpha_m$: old worlds fulfill none of $t_m,t'_m,f_m,f'_m$ due to (3), and if $w \vDash s_m$, then $w \vDash \alpha_m$ by the construction shown in Figure~\ref{fig:af-construction}.

To see that still $(\calK_m,w)\vDash \alpha_i$ for $1 \leq i < m$, suppose $(\calK_m,w)\nvDash \alpha_i$ for the sake of contradiction. Since $w \in V_m(p) \Leftrightarrow w \in V_{m-1}(p)$ for all $p \in \PS \cap \SF{\alpha_i}$, this implies $(\calK_m,w)\nvDash \AF \xi$ and $(\calK_{m-1},w)\vDash \AF \xi$, for some $\AF \xi \in \SF{\alpha_i}$.
So let $(\calK_m,\pi')\vDash \G \neg \xi$ for a path $\pi' = (w, \ldots)$.
This path cannot visit $U_g$ or $U_\ell$, since $t_i,t'_i,f_i,f'_i,\neg b_i$ and consequently $\xi$ are true in every world of $U_g\cup U_\ell$. For this reason, it must already hold $(\calK_{m-1},\pi'')\vDash \G\neg \xi$ for some subpath $\pi'' = (w,\ldots)$ of $\pi'$ through $\calK_{m-1}$. But this contradicts $(\calK_{m-1},w)\vDash \AF \xi$.

Next, we consider $\beta_i$ for $i \in [m+1]$. Trivially $\beta_{m+1} = (b_{m+1} \imp \EG b_{m+1})$ holds in all old worlds, since $b_{m+1}$ occurs only in new worlds. For $i \leq m$, we apply the induction hypothesis: either $w \nvDash b_i$, or there is an $R_{m-1}$-path $\pi'=(w,\ldots)$ such that $(\calK_{m-1},\pi')\vDash \G b_i$. But by property (2), $\pi'$ never visits a world where $s_m$ holds. Consequently, it is still an $R_m$-path and witnesses $(\calK_m,w)\vDash \EG b_{i}$.

With respect to $(\bigwedge_{i=1}^{m-1}\gamma_i)$, the new worlds are "transparent" as follows.
Whenever $(\calK_{m-1},\pi')\vDash \G \neg p$ for $p \in \{t'_i,f'_i\}$, then an $R_m$-path $\pi''$ can be obtained from $\pi'$ by replacing deleted edges $(\pi[j],\pi[j+1])$ by the corresponding steps through the new worlds.
Then still $(\calK_m,\pi'') \vDash \G \neg p$, as $t'_i$ or $f'_i$ are false in all new worlds.

Finally, $w\vDash \gamma_m$ holds for all old worlds $w$ because there is always an $R_m$-path $\pi' = (w, \ldots)$ such that $(\calK_m,\pi')\vDash \G (\neg t'_m \land \neg f'_m)$. Given $w = \pi[j]$, we construct $\pi'$ as follows. If a minimal $k \geq j$ exists such that $\pi[k]\vDash s_m$, let $\pi'\dfn (\pi[j],\pi[j+1],\ldots,\pi[k],u^k_g,u^k_g,\ldots)$. Otherwise $\pi' \dfn \pi_j$ contains only old worlds and, by the induction hypothesis (3), is the desired path.
\end{proof}
The reduction maps any true qbf $\varphi$ with prefix $\forall x_1 \cdots \forall x_n$ to a CTL formula $\varphi^*$ of length $\bigO{n}$, and with any model of $\varphi^*$ having extent at least $2^n$.

\begin{corollary}\label{cor:model-size-af}
$\B_k(\AF)$ and $\B(\AF)$ have optimal model size and extent
$2^\bigTheta{n}$ for all $k \geq 2$.
\end{corollary}

It may seem surprising that $\AF$ can enforce a single exponentially long path,
whereas this is not possible with the LTL-operators $\F$ and $\G$.
The reason for this is twofold:
On the one hand, $\F$ operators enjoy a certain ``order invariance'':
With respect to a formula $\varphi$ and a model, the set of fulfilled subformulas of the form $\G\psi \in \SF{\varphi}$ can only grow along a given path. For this reason, every path has a finite prefix after which no new $\G$-formulas are imposed such that the order of fulfillment of $\F$ does not matter anymore.
On the other hand, all $\G$-formulas occurring on a path
affect that path due to the lack of branching and must not contradict. With $\EG$, paths may however ``branch off'' arbitrarily. Both properties are used by Sistla and Clarke to show the polynomial
model property of certain LTL fragments \cite{SC85}, while conversely the absence of both
properties is crucial for the proof presented here.

\subsection{The AG fragment}\label{sec:ag}

In terms of computational complexity, the $\AG$ fragment is well-understood: it is equivalent to the modal logic $\mathsf{S4D}$, \ie, on transitive, reflexive, serial frames.

\begin{proposition}\label{prop:agef-pspace-complete}
$\SAT(\B(\AG)) \in \PSPACE$.
\end{proposition}
\begin{proof}
A $\B(\AG)$-formula is satisfiable if and only if it has a serial, reflexive, and transitive model. On such structures, however, $\AG$ is equivalent to the modal "Box" operator $\Box$.
Therefore the $\mathsf{S4}$-satisfiability algorithm given by Ladner
\cite{ladner77} provides the desired result, with little modifications
to respect seriality.
\end{proof}

Next, we will improve the lower bounds for this logic, in particular we show that it already holds for temporal depth two.
We refine the classical proof which reduces from $\TQBF$ to $\sfS4\sfD$-satisfiability by expressing the existence of proof trees in modal logic.
While the idea is roughly the same as in the $\AF$ case---force a
Kripke structure to carry up to $2^n$ different propositional
assignments---the implementation fundamentally differs due to the
different semantics of $\AF$ and $\AG$. When using the first operator,
we must use a single exponentially long path, and with the second we
have an exponentially branching tree with linear depth. We will later
see a linear upper bound for the optimal model extent as well.

\begin{theorem}\label{thm:qbf-to-agef}
$\SAT(\B_2(\AG))$ is $\PSPACE$-hard.
\end{theorem}
\begin{proof}
Let $\varphi = Q_1 x_1 \ldots Q_n x_n \psi$ be a qbf. We reduce $\varphi$ to the formula $\varphi^*$, defined as follows:
\begin{align*}
\varphi^* \dfn &\;y_0 \sland
\AG \Big( \big((y_n \lor z_n)\imp \psi\big) \sland \bigand_{i=1}^{n} \alpha_i \Big)\text{,}\\
\alpha_i \dfn &\; \Big((y_{i-1} \slor
z_{i-1})\rightarrow ( \EF y_i
\circ_i \EF z_i ) \Big)
\sland \Big( y_i \rightarrow \AG
x_i\Big) \sland
\Big( z_i \rightarrow \AG
\neg x_i \Big)\text{,}
\end{align*}
where $\circ_i \dfn \land$ if $Q_i = \forall$, and $\circ_i \dfn \lor$ if $Q_i = \exists$, and for all $0 \leq i \leq n$, the symbols $y_i, z_i$ are fresh propositions. Clearly the formula is logspace-constructible.
Intuitively, as soon as $y_i$ is true in some world $w$, $x_i$ shall be true in all worlds reachable from $w$. Analogously, if $z_i$ holds, then $x_i$ shall be false in all reachable worlds.\end{proof}

To prove the correctness of the reduction, we again use a lemma for each direction.

\begin{lemma}
If $\varphi^*$ is satisfiable, then $\varphi$ is true.
\end{lemma}
\begin{proof}
Let $(\calK,r) \vDash \varphi^*$, $\calK = (W,R,V)$. \Wloss $(\calK, r)$ is $R$-generable. We prove that $\calK$ simulates a proof tree for $\varphi$, similarly as in Lemma~\ref{lem:proof-tree-in-af}.
Let $X_0 \dfn \{r\}$ and, for $1 \leq m \leq n$, let
$X_m \dfn \Set{ w \in V(y_m) \cup V(z_m) | \exists w' \in X_{m-1} : w'R^*w }\text{.}
$
The meaning of the set $X_m$ is that the truth of $x_1,\ldots,x_m$ is already "fixed" in $w \in X_m$, in the sense that its assignment to $x_1,\ldots,x_{m-1}$ is recursively determined by $w$ being reachable from a world in $X_{m-1}$, and $x_m$ being selected from satisfying either $y_m$ or $z_m$.

We will ascertain that the following tree $T = (\Theta,E)$ is a proof tree of $\varphi$:
\begin{align*}
    \Theta &\dfn \Set{ \theta \colon \{x_1,\ldots,x_m\} \to \{0,1\} | 0 \leq m \leq n, \exists w \in X_m : w \vdash \theta}\\
    E &\dfn \Set{(\theta, \theta') \in \Theta^2 | \begin{array}{l}\theta \colon \{x_1,\ldots,x_{m-1}\} \to \{0,1\},
    \theta' = \theta^{x_m}_b \text{ for some } b \in \{0,1\}\text{,}\\
    \text{and }\exists w \in X_{m-1}, \exists w' \in X_m, w \vdash \theta, w'\vdash \theta', wR^*w'\end{array}}
\end{align*}

$\Theta$ contains the empty Boolean assignment, as $r \in X_0$. Whenever $\theta \colon \{x_1,\ldots,x_{m-1}\} \to \{0,1\}$ is in $\Theta$, then it agrees with some $w \in X_{m-1}$ by definition of $\Theta$. Since $w \in X_{m-1}$, all worlds reachable from $w$ must have the same truth values for $x_1,\ldots,x_{m-1}$ as $w$. Assuming $Q_m = \exists$, and by $\alpha_m$, there is a world $w' \in X_m$ that agrees either with $\theta^{x_m}_0$ or with $\theta^{x_m}_1$.
Conversely, $Q_m = \forall$, then two worlds $w',w'' \in X_m$ agreeing with $\theta^{x_m}_0$, $\theta^{x_m}_1$ exist. Ultimately, $\theta^{x_m}_0$, $\theta^{x_m}_1$, or both are in $\Theta$ and children of $\theta$ in $T$, depending on $Q_m$.

If a leaf assignment $\theta \colon \{x_1,\ldots,x_n\} \to \{0,1\}$ is in $\Theta$, then it agrees with some world $w \in X_n$. But since $(y_n\lor z_n)\imp \psi$ holds in all worlds of $\calK$, it also follows $w \vDash \psi$.
By these arguments, the conditions (1)--(3) of Definition~\ref{def:witness-seq} are true in $T$, such that $T$ ultimately contains a proof tree of $\varphi$.
\end{proof}

\begin{lemma}
If $\varphi$ is true, then $\varphi^*$ is satisfiable.
\end{lemma}
\begin{proof}
Suppose that $\varphi$ is true and that accordingly $T = (\Theta, E)$ is a proof tree of $\varphi$. We define a Kripke structure $\calK = (\Theta,E,V)$ such that $(\calK,\theta_0)\vDash \varphi^*$, where $\theta_0\in\Theta$ is the empty assignment.

For $i \in [n]$, let
\begin{align*}
    V(x_i) &\dfn \Set{ \theta \in \Theta | \theta(x_i) \text{ is defined and }\theta(x_i) = 1  }\\
    V(y_i) &\dfn \Set{ \theta \in \Theta | \theta \colon \{x_1,\ldots,x_i\} \to \{0,1\}, \theta(x_i) = 1 }\\
    V(z_i) &\dfn \Set{ \theta \in \Theta | \theta \colon \{x_1,\ldots,x_i\} \to \{0,1\}, \theta(x_i) = 0 }
\end{align*}
and otherwise $V(y_0) \dfn \{ \theta_0 \}, V(z_0)\dfn \emptyset$.
By this construction, and since any assignment $\theta \in \Theta$ defined on $\{x_1,\ldots,x_n\}$ satisfies $\psi$, it follows $(\calK,\theta_0)\vDash y_0 \land \AG \big((y_n \lor z_n) \imp \psi\big)$.
On the other hand, for $i \in [n]$, the truth of $\alpha_i$ is easy to verify by the definition of proof trees.
\end{proof}

\begin{corollary}\label{cor:ag-model-lower}
$\B_k(\AG)$ and $\B(\AG)$ have model size lower bound $2^{\bigOmega{n}}$ and extent lower bound $\bigOmega{n}$ for all $k \geq 2$.
\end{corollary}

The next result is the matching upper bound for model extent.
For this we introduce the notion of \emph{quasi-models}.
The crucial difference to a model is that we do not need to talk about
\emph{truth} of a subformula, but rather only whether or not a subformula or
its negation is \emph{necessitated} in a specific world at all. The idea is that every necessary formula must be true, but not vice versa. This
approach is well-known in literature for establishing
upper bounds for model size, often together with filtration
techniques. Related notions are \emph{Hintikka structures}, \emph{pseudo-models} or \emph{tableaux}, see also Allen Emerson and Halpern \cite{emersonTemporal,halpern}.

Let $\negg \psi \dfn \xi$ if $\psi = \neg \xi$ for some $\xi$, and let $\negg \psi \dfn \neg \psi$ otherwise.

\begin{definition}[Closure]
Let $\varphi \in \B(C)$ for a base $C$.
The \emph{closure} $cl(\varphi)$ of $\varphi$ is the smallest set
for which holds:
\begin{itemize}
  \item $\varphi \in cl(\varphi)$.
  \item if $Q O\psi \in cl(\varphi)$ for unary $QO \in \TL$, then $\{\psi, \overline{Q}\,\overline{O} \negg \psi \} \subseteq cl(\varphi)$,
  \item if $Q[\psi O \xi] \in cl(\varphi)$ for binary $QO \in \TL$, then $\{\psi,\xi,\overline{Q}\,[\negg\psi \overline{O} \negg\xi]\} \subseteq cl(\varphi)$,
\item if $f(\psi_1, \ldots, \psi_n) \in cl(\varphi)$, $f \in C$,
then $\psi_1, \ldots, \psi_n \in cl(\varphi)$,
\item $\psi \in cl(\varphi)$ iff $\negg\psi \in cl(\varphi)$, that is, for every formula in $cl(\varphi)$ also a formula equivalent to the negation is in $cl(\varphi)$.
\end{itemize}

For a set $\Phi$ of formulas, define $cl(\Phi) \dfn \bigcup_{\varphi \in \Phi} cl(\varphi)$.
\end{definition}

The closure $cl$ is similar to the \emph{Ladner-Fischer closure}
defined for PDL \cite{fischer_propositional_1979}.
Note that not necessarily $\neg \in C$, but $cl(\varphi) \subseteq \B(C \cup \{\neg\}, T)$ if $\varphi \in \B(C,T)$.

\begin{definition}[Quasi-models]
Let $\varphi \in \B(C)$. A \emph{quasi-model} of $\varphi$ is
then a tuple $\calQ = (W, R, L)$, where $(W, R)$ is a serial Kripke frame, and
$L \colon cl(\varphi) \to \powerset{W}$ is the
\emph{extended labeling function} and obeys the following conditions:
\begin{enumerate}[label=(Q\arabic*)]
    \item if $f \in C$, $\xi = f(\psi_1,\ldots,\psi_{\arity{f}})$ and $w \in L(\xi)$ resp.\ $w \in L(\negg \xi)$, then $\exists \vec b \in \{0,1\}^n$ such that
    \begin{itemize}
        \item $f(\vec{b}) = 1$ resp.\ $f(\vec b) = 0$
        \item $\forall i \in [n]$, $b_i = 1$ implies $w \in L(\psi_i)$ and $b_i = 0$ implies $w \in L(\negg \psi_i)$,
    \end{itemize}

    \item $L(\psi) \cap L(\negg \psi) = \emptyset$ for all $\psi \in cl(\varphi)$,
  \item if $w \in L(\neg QO \psi)$ for unary $QO
  \in \TL$, then $w \in L(\overline{Q}\,\overline{O} \negg \psi)$,
  \item if $w \in L(\neg Q[\psi O \xi])$ for binary $QO
  \in \TL$, then $w \in L(\overline{Q}[\negg\psi\overline{O}\negg\xi])$,
  \item if $w \in L(\E\psi)$ \normalbraces{$w \in L(\A \psi)$}, then for some (all)
  paths $\pi \in \Pi(w)$:
  \begin{itemize}
    \item if $\psi = \X \beta$, then $\pi[1] \in L(\beta)$
    \item if $\psi = \F\beta$ \normalbraces{$\G\beta$}, then $\pi[i] \in L(\beta)$ for some (all) $i \geq 0$,
    \item if $\psi = \beta\U\xi$ \normalbraces{$\beta\RLS\xi$}, then
    for some (all) $i \geq 0$ it is $\pi[i] \in L(\xi)$ and (or)
    $\pi[j]\in L(\beta)$ for all (some) $j < i$,
  \end{itemize}
  \item $L(\varphi) \neq \emptyset$.
\end{enumerate}

The properties (Q1)--(Q4) are the
\emph{local quasi-label conditions}.
\end{definition}

As in usual Kripke structures, we sometimes call the set
$\Set{ \psi \in cl(\varphi) | w \in L(\psi)}$ the \emph{quasi-labeling} of $w$, or just \emph{labeling} of $w$ if the context is clear, and for any formula $\psi$ in the above set we say that $\psi$ is \emph{labeled in $w$}.

\smallskip

Models and quasi-models are equivalent in the following sense:

\begin{proposition}\label{prop:quasi-to-model}Let $\varphi \in \B$.
\begin{enumerate}
\item If $(W,R,V,w)$ is a model of $\varphi$, then $(W,R,L_\varphi)$ is a quasi-model of $\varphi$, where $L_\varphi(\psi) \dfn \Set{ u \in W | (W,R,V,u) \vDash \psi}$ for all $\psi \in cl(\varphi)$.
\item If $(W,R,L)$ is a quasi-model of $\varphi$, then for all $w \in L(\varphi)$, $(W,R,V_L,w)$ is a model of $\varphi$, where $V_L(p) \dfn L(p)$ for all $p \in \PS \cap \SF{\varphi}$ and $V_L(p) \dfn \emptyset$ otherwise.
\end{enumerate}
\end{proposition}
\begin{proof}
Induction on the length of the formula.
\end{proof}

\begin{definition}Let $\calM = (W, R, V, w_0)$ be a rooted Kripke structure. The \emph{tree unraveling} $\calM^T$ of $\calM$ is defined as $\calM^T = (W^T, R^T, V^T, (w_0))$, where
$W^T$ is the set of all prefixes of paths $\pi \in \Pi(w_0)$, the relation
\[
R^T \dfn \Set{ \big((w_0,\ldots,w_n),(w_0,\ldots,w_n,w_{n+1})\big) | (w_0,\ldots,w_n) \in W^T, w_nRw_{n+1}}
\]
is the maximal proper path prefix relation, and $(w_0,\ldots,w_n) \in L^T(p)$ if and only if $w_n \in L(p)$.
\end{definition}

In what follows, when a Kripke frame is called a \emph{tree}, then the meaning is that it forms a rooted, directed tree where every edge points away from the root.
Clearly, the underlying Kripke frame of $\calM^T$ forms an infinite tree.

\begin{proposition}[\cite{emersonTemporal}]\label{prop:tree-model}
If $\varphi\in \B$ and $\calM$ is a model of $\varphi$, then $\calM^T$ is a model of $\varphi$.
\end{proposition}

With (infinite tree) quasi-models in the toolbox, we are now able to prove the upper bound in model extent for the $\AG$ fragment. The idea is to "greedily" construct a new model, in the sense that $\EF$-subformulas are always fulfilled in immediate successor worlds.

\begin{theorem}\label{thm:model-size-ag}
For any base $C$, $\B(C,\AG)$ has model extent upper bound $\bigO{n}$.
\end{theorem}
\begin{proof}
Let $\varphi \in \B(\AG)$ be satisfiable, and $(\calT,r) = (W, R, L, r)$ an infinite tree quasi-model obtained from the unraveling of a model of $\varphi$.
\Wloss{} $w \in L(\AG \xi)$ implies $u \in L(\AG\xi)$ for all for all $\AG \xi \in cl(\varphi)$, $w \in W$, and $R$-successors $u$ of $w$.

For $w \in W$, define
\[
\calF(w) \dfn \Set{ \xi \in cl(\varphi) | w \in L(\EF\xi) \setminus L(\xi)}\text{,}
\]
\ie, the set of unfulfilled $\EF$-formulas labeled in $w$.
Analogously, let
\[
\calG(w) \dfn \Set{ \xi \in cl(\varphi) | w \in L(\AG \xi)}\text{.}
\]
We introduce a \emph{candidate set} $\calC_\xi(w) \subseteq W$ for each $w\in W$ and $\xi \in cl(\varphi)$. Let $\calC_\xi(w) \dfn \Set{ u \in W | u \in L(\xi), wR^*u }$.
For $\xi \in \calF(w)$, the set $\calC_\xi(w)$ is non-empty; it contains reachable worlds $u$ that witness the truth of $\EF\xi$ in $w$.
We define a subset
\[
\calC^\mathrm{max}_\xi(w) \dfn \Set{ u \in \calC_\xi(w) | \forall u' \in \calC_\xi(w), u' \neq u, \size{\calG(u)} \geq \size{\calG(u')} }\text{,}
\]
which is the restriction to candidates $u$ that are \emph{maximal} with respect to the number of their labeled $\AG$-formulas.
The new edge relation $E$ based on $\calC^\mathrm{max}_\xi(w)$ is
\[
E \dfn \Set{ (w,u) \in W \times W | \xi \in \calF(w), u \in \calC^\mathrm{max}_\xi(w)}\text{.}
\]
To ensure the seriality of the new model, we furthermore require reflexive edges $E^\mathrm{ref} \dfn \Set{ (w,w) \in W \times W | \calF(w) = \emptyset }$.

We define the quasi-model $\calT' \dfn (W',E \cup E^\mathrm{ref}, L')$, where
$W' \dfn \Set{ w \in W | rE^*w}$ is the restriction of $W$ to worlds reachable from $r$.
Similarly, let $L'(\xi) \dfn L(\xi) \cap W'$ for all $\xi \in cl(\varphi)$.
It is straightforward to check that $\calT'$ is a quasi-model for $\varphi$.
$\calT'$ is not necessarily finite; the rest of the proof describes how to reduce $\calT'$ to a finite model with linear extent.

\smallskip

First, consider a mapping $J_w$ from proper successors of $w$ to $\calF(w)$ such that $J^{-1}_w(\xi) \in L'(\xi)$ for all $\xi \in \calF(w)$. We can think of $J_w$ as the \emph{justifications}, in the sense that every successor is responsible for a $\EF$-formula in $w$. \Wloss, $J_w$ is a bijection (clone successors until there are enough, and delete unused ones). As all worlds $u$ in $\calT'$ have at most one proper predecessor $w$, simply write $J(u)$ for $J_w(u)$.

Next, we show that justifications may only repeat on a path if the corresponding $\calG$-sets are equal.
Let $\pi = (u_0, u_1, \ldots)$ be a path through $\calT'$.
For every world $u_i$ with $i \geq 1$, there is a justification $\xi = J(u_i)$.
Assume $J(u_i) = J(u_j) = \xi$ for
some $1 \leq i<j$. We show that $\calG(u_i) =\calG(u_j)$. Clearly $\calG(u_i) \subseteq \calG(u_j)$ follows from the assumption we made at the beginning, since $u_iR^*u_j$. If however $\calG(u_i)\subsetneq \calG(u_j)$, then $u_i \notin \calC^\mathrm{max}_{\xi}(u_{i-1})$, contradiction.

We now transform $\calT'$ to a finite quasi-model $\calM$ as follows:
While there is a \emph{long} path $\pi$, \emph{furl} that path.
A path is \emph{long} if it visits more than $|\{\EF \xi \in \SF{\varphi}\}|$ distinct worlds besides $r$.
To \emph{furl} a path $\pi$, choose the
minimal $j$ such that $J(\pi[j]) = J(\pi[i])$ for some $i < j$.
Such an $j$ must exist if $\pi$ is long.
The world $\pi[j]$ can then be
replaced by a back edge from $\pi[j-1]$ to $\pi[i]$ without violating
any quasi-label condition: no $\AG$ is violated, as $\calG(\pi[i]) = \calG(\pi[j])$, and no $\EF$ is violated, as every world $\pi[j]$ needs only to satisfy its justification.
As $\calM$ then has no long paths, $\varphi$ has a model with extent $\bigO{\varphi}$.
\end{proof}

\subsection{The AX fragment}\label{sec:ax}

The $\AX$ fragment of temporal logic is, similar to $\AG$, well known from the context of modal logic. The following theorems are adaptations of some of its properties.

\begin{theorem}\label{theorem:model-size-ax-lower}
$\B(\AX)$ has a model size lower bound $2^{\Omega\left(\sqrt{n}\right)}$ and extent lower bound $\bigOmega{n}$.
$\B_k(\AX)$ has a model size lower bound $n^{\bigOmega{k}}$ and extent lower bound $k$.
\end{theorem}
\begin{proof}
The temporal depth as lower bound for model extent is straightforward.
For the size lower bound, we enforce a large model with a standard approach (see also \cite{Marx2007139}). Let
\[
\psi^m_i \dfn \left[ \bigand_{j=0}^{m-1} \EX
\vec{c_i}(j) \right] \sland \bigand_{s=1}^{i-1}
\bigand_{t = 0}^{\ceil{\log m}} (p_{s,t}
\rightarrow \AX p_{s,t}) \sland (\neg p_{s,t} \rightarrow \AX \neg p_{s,t})\text{,}
\]
where $\vec{c_i}(j)$ is a conjunction of $\log m$ literals of propositions $p_{i,1}, \ldots, p_{i,\log m}$, such that $\vec{c_i}(j)$ represents the binary value $j$.
Then the satisfiable formula
\[
\varphi_{m,k} \dfn \psi^m_1 \land \AX(\psi^m_2 \land \AX(\ldots (\AX \psi^m_k) \ldots ))
\]
has length $\bigO{k^2 \cdot m \log m}$ and temporal
depth $k$, but no model with less than $m^k$ worlds.
For fixed $m \geq 2$, consequently $\varphi_{m,k}$ has length $\bigO{k^2}$ and only models of size $\geq 2^k$.
Conversely, for fixed $k$, $\varphi_{m,k}$ has length $\bigO{m \ln m} \subseteq \bigO{m^2}$, and only models of size $\geq m^{k} = (m^2)^{\frac{k}{2}}$. As a result, we obtain a family of formulas of size $n$ with models of size at least $n^{\bigOmega{k}}$.
\end{proof}

\begin{proposition}\label{prop:model-size-ax}
For any base $C$ and any $k$, $\B_k(\AX)$ has a model size upper bound $n^{\bigO{k}}$ and extent $k$. $\B(\AX)$ has a model extent upper bound $n$.
\end{proposition}
\begin{proof}
Clearly the temporal depth is an upper bound for the extent. Every world in a model of $\varphi$ requires at most $\max\{1,\ell\,\}$ successors, where $\ell$ is the number of $\EX$-subformulas in $\varphi$.
The model size for $\varphi \in \B_k(\AX)$ follows, as $\sum_{i = 0}^k \size{\varphi}^i \in\size{\varphi}^{\bigO{k}}$.\phantom{$\Box$}
\end{proof}

The gap between the upper bound $2^{\bigO{n}}$ and lower bound $2^{\Omega\left(\sqrt{n}\right)}$ can be closed by choosing a different encoding for modal formulas. In more succinct encodings, for instance \emph{modal circuits} (see Hemaspaandra, Schnoor, and Schnoor~\cite{hss10}), a lower bound of $2^{\bigOmega{n}}$ can be achieved.

\begin{proposition}\label{prop:ax-k-npc}
For all $k \geq 0$, $\SAT(\B_k(\AX))$ is $\NP$-complete.
\end{proposition}
\begin{proof}
The upper bound follows from the previous theorem: Guess a
satisfying model of polynomial size and verify it in polynomial time, since CTL model checking is in $\P$ \cite{clarke_automatic_1986}. The lower bound
holds as $\B_0$ is nothing else than propositional logic, for which the
satisfiability problem is already $\NP$-complete \cite{cook_complexity_1971}.
\end{proof}

A complete classification of the computational complexity of modal satisfiability (in the case of unbounded modal depth and arbitrary Boolean bases) was given by Hemaspaandra et al.\ \cite{hss10}.
Since serial modal logic $\KD$ with Boolean base $C$ is equivalent to $\B(C,\AX)$, clearly the next theorem follows:

\begin{theorem}[\cite{hss10}]\label{thm:ax-pspace-completeness}
Let $C$ be a base such that $\CloneSE \subseteq [C]$. Then $\SAT(\B(C, \AX))$ is $\PSPACE$-complete.
\end{theorem}

\subsection{The AF AX fragment}\label{sec:afax}

The next part establishes the matching upper bounds for the fragment
with both $\AX$ and $\AF$. It requires some technical work; we show
$\PSPACE$ membership by constructing a canonical \emph{balloon model}. It has a special form that allows to non-deterministically guess and
verify it on-the-fly, namely it is "pseudo-acyclic":
it almost resembles a tree, except that its branches are closed into
cycles. This poses a strong restriction to possible back-edges, and allows to guess
such a model in a depth-first search manner using polynomial space.

We require several auxiliary definitions:

\begin{definition}[Ultimately periodic path]
  A path $\pi$ of the form
  \[
  \pi = (w_1, \ldots, w_i, w_{i+1}, \ldots, w_{i+k}, w_{i+1}, \ldots)\text{,}
  \]
  where $i \geq 0,k\geq 1$,
  is called \emph{ultimately periodic}.
  It consists of a finite \emph{prefix} that visits every world at most once, followed by an infinite repetition of a finite, non-empty \emph{cycle}.
The \emph{length} of $\pi$ is $i + k$, \ie, the sum of the
  lengths of its prefix and its cycle.
\end{definition}

\begin{definition}[Balloon path]
  Let $F = (W,R)$ be a finite Kripke frame with exactly one predecessor-free world $r \in W$ (its \emph{root}), and all other worlds reachable from $r$. We call $F$ a \emph{balloon path} if there is exactly one $R$-path $\pi$ with origin $r$ (then $\pi$ must be ultimately periodic with non-empty prefix, as $F$ is assumed finite). The \emph{length} of $F$ is then simply the length of $\pi$.
\end{definition}

\begin{definition}[Balloon frames]
A Kripke frame $F = (W,R)$ is called a \emph{balloon frame of level $m$ and
length at most $n$}, provided that
\begin{itemize}
  \item for $m = 0$, $F$ is a balloon path of length at most $n$.
  \item for $m > 0$, $F$ is the union of Kripke frames $P, F_1,\ldots,F_k$, where $P$ is a balloon path of length at most $n$, and for all $i, j \in [k]$,
  \begin{itemize}
    \item $F_i$ is a balloon frame of length $n$ and level at most $m - 1$,
    \item for $i \neq j$, $F_i$ and $F_j$ are disjoint except their roots,
    \item $F_i$ and $P$ are disjoint except that the root of $F_i$ must be a world of $P$.
  \end{itemize}
\end{itemize}
\end{definition}

Intuitively, $F$ is constructed by taking a balloon path of length at most $n$, and appending to each world $u$ a finite number of balloon structures of level at most $m-1$ and length at most $n$. \emph{Appending} here means identifying each of their roots with $u$ such that they have no other worlds in common with each other.

Such a structure with bounded level $m$ has some useful properties. For instance, every path must visit at most one balloon frame of level $m, m-1, \ldots, k$ for some $k \geq 1$, and then stay forever in that one with level $k$.

\smallskip

If $(W,R)$ is a balloon structure and $(W,R,L)$ is a quasi-model of a formula
$\varphi$, then $\calM = (W,R,L)$ is a \emph{balloon quasi-model} of $\varphi$.

The first step towards finding a balloon quasi-model is to identify ultimately periodic paths as witnesses for $\E$-formulas. Here, a path with origin $w$ \emph{witnesses} a formula $\E\gamma$ labeled in $w$ when $\pi[1] \in L(\delta)$ (if $\gamma = \X \delta$) resp.\ when $\Set{\pi[i] | i \geq 0} \subseteq L(\delta)$ (if $\gamma = \G \delta$).

\begin{lemma}\label{lem:ulti-path}
Let $\calM = (W, R, L)$ be a finite quasi-model of $\varphi \in \B$. Assume $\gamma \in cl(\varphi)$ is a formula of the form $\EX \delta$ or $\EG \delta$. If $w \in L(\gamma)$, then there is a path $\pi \in \Pi(w)$ witnessing $\gamma$ such that $\pi$ is ultimately periodic and of length at most $\size{W}$.
\end{lemma}
\begin{proof}
Let $\pi \in \paths(w)$ be the path through $\calM$ that witnesses the truth of $\gamma$. Let $j$ be minimal such that $\pi[j] = \pi[j']$ for some $j' < j$, \ie, $\pi[j]$ is the first world visited twice on $\pi$. Obviously $j$ must exist.

Then the path $\pi' \dfn (\pi[0],\ldots,\pi[j'],\ldots,\pi[j-1],\pi[j'],\ldots)$ is ultimately periodic and of length at most $\size{W}$. Furthermore, if $\gamma = \EX \delta$, then $\pi'[1] \in L(\delta)$, since always $\pi'[1] = \pi[1]$, and if $\gamma = \EG \delta$, then $\Set{\pi'[i] | i \geq 0} \subseteq \Set{\pi[i] | i \geq 0 }\subseteq L(\delta)$.
\end{proof}

Next, we restrict the possible selection of witness paths even further to obtain a balloon-like structure.
Formally, every finite quasi-model $(W,R,L)$ of $\varphi$ has a \emph{choice function} $f \colon W \times cl(\varphi) \to \bigcup_{w \in W}\Pi(w)$ with respect to satisfaction of labeled $\E$-formulas. For $\E\gamma \in cl(\varphi)$ and $w \in L(\E\gamma)$, $f(w,\E\gamma)$ is defined as a path $\pi\in\Pi(w)$ witnessing $\gamma$.
We call such a choice function $f$ \emph{normal} if:
\begin{itemize}
    \item $f$ is injective,
    \item for any path $\pi$ in the image of $f$, $\pi$ is ultimately periodic, and worlds having two or more $R$-predecessors may only occur as the root of $\pi$ or in its cycle,
    \item for any two paths $\pi,\pi'$ in the image of $f$, $\pi_{\geq 1}$ and $\pi'_{\geq 1}$ are disjoint (but possibly $\pi[0]$ is an element of $\pi'$ or vice versa).
\end{itemize}

Intuitively, a normal choice function in a balloon quasi-model means that witness paths with origin $w$ always branch into a new balloon of shallower level and root $w$; moreover, for every $\E$-formula a distinct balloon is attached.

\smallskip

For easier argumentation, in what follows we assume \wloss that subformulas of $\varphi$ containing temporal operators can occur only once in $\varphi$. Formally,
\[
(\psi_1\notin \SF{\psi_2} \land \psi_2\notin \SF{\psi_1}) \Rightarrow cl(\psi_1) \cap cl(\psi_2) \subseteq \B_0\text{.}
\]
In other words, if two formulas have no common subformulas, then they also have no common elements in their closure except propositional formulas.\footnote{By, for instance, introducing enough copies $O', O'', \ldots$ of any temporal operator $O$.}

For the rest of the subsection, we introduce a new quasi-label condition that can be assumed without loss of generality.
\begin{enumerate}[label=(Q\arabic*),start=7]
        \item If $w \in L(\AF \psi) \setminus L(\psi)$, then \begin{itemize}
    \item $u \in L(\psi)$ or $u \in L(\AF\psi)$ for all successors $u$ of $w$, and
    \item $w \notin L(\psi')$ for all $\psi' \in cl(\psi) \setminus \B_0$.
    \end{itemize}
\end{enumerate}

The first condition is known as the \emph{fixpoint characterization} of $\AF$ and is used in the upcoming technical proof.
The second condition is a sort of "negative downwards closure": if $\AF\psi$ is not fulfilled in $w$, then certainly none of its subformulas are required in $w$. This can be assumed due to the uniqueness of subformulas explained above.

These conditions are crucial later to construct a model in finitely many steps.
Note that the second condition is exactly the failing point for the operators $\AG$, $\AU$ and $\A\RLS$, as they \emph{do} always necessitate labeling their subformulas in $w$, with or without being fulfilled.

\begin{lemma}[Balloon lemma]\label{lem:balloon-lemma}
If $T \subseteq \{ \AF, \AX \}$, $C$ is a base, and $\varphi \in \B(C, T)$, then $\varphi$ is satisfiable if and only if it has a balloon quasi-model of level $\bigO{\size{\varphi}}$ and length  $2^\bigO{\size{\varphi}}$ that has a normal choice function.
\end{lemma}
\begin{proof}
Let $\varphi$ be satisfiable. From any balloon quasi-model we can obtain a model of $\varphi$ by Proposition~\ref{prop:quasi-to-model}.
Conversely, by Theorem~\ref{thm:upper-size-all} and Proposition~\ref{prop:quasi-to-model}, $\varphi$ has a quasi-model $\calM = (W, R, L)$ of size $2^\bigO{\size{\varphi}}$.

We construct a balloon quasi-model $\calT = (W', R', L')$ in stages as follows.
Select a world $r \in L(\varphi)$ as the root of $\calT$, and \wloss assume that at least one $\E$-formula is labeled in $r$. Let $r \in W'$. We will subsequently add more worlds to $W'$, connected by $R'$-edges, and define their quasi-label given by $L'$ accordingly, such that $\calT$ eventually is a balloon quasi-model of $\varphi$.

In the construction, define $cl(w) \dfn \bigcup\Set{ cl(\psi) | \psi \in cl(\varphi), w \in L'(\psi) }$, the union of the closures of all formulas labeled in $w \in W'$.
To keep track of the balloons of level $m-1$ emanating from worlds on a balloon of level $m$, we further introduce a \emph{level function} $\ell \colon W' \to \N$. Set $\ell(r) \dfn \size{cl(\varphi)}$. Moreover, we will define the required normal choice function $f$ during the construction.

\smallskip

Now, for all formulas $\E\gamma \in cl(\varphi)$ and all worlds $w \in L'(\E\gamma)$ with $\ell(w) > 0$, do the following. If $f(w,\E\gamma)$ is not yet defined, then select a path $\pi$ from $\calM$ with origin $w$ according to Lemma~\ref{lem:ulti-path} to satisfy $\E\gamma$. $\pi$ is ultimately periodic with length $2^\bigO{\size{\varphi}}$. Append a copy $\pi'$ of this path to $w$. If the appended path happens to have an empty prefix, \ie, $\pi = (v_1, \ldots, v_m, v_1, \ldots, v_m, \ldots)$, then use as prefix a copy $(v_1',\ldots,v_m')$ of the cycle, to ensure that the appended worlds form a balloon path. Set $\ell(u) \dfn \ell(w) - 1$ for each such appended world $u$, and define the choice function as $f(w,\E\gamma)\dfn \pi'$.
Afterwards, assuming $\gamma = \X\delta$ or $\gamma = \G\delta$, leave only formulas $\psi$ labeled in $\pi'_{\geq 1}$ such that, for all $j \geq 1$,
\begin{align*}
cl(\pi'[j]) \subseteq cl(\delta) &\cup \bigcup\Set{ cl(\alpha) | w\in L'(\AX\alpha)}\\
&\cup \bigcup\Set{ cl(\AF\alpha) | w \in L'(\AF\alpha) \setminus L'(\alpha) }\text{.}
\end{align*}

It is straightforward to check that this does not violate any quasi-label condition.
This construction terminates and leaves a balloon quasi-model $\calT$ of level $\size{cl(\varphi)}$ and length at most $2^\bigO{\size{\varphi}}$, and having a normal choice function $f$. To prove that $\calT$ is indeed a quasi-model of $\varphi$, we have to show that all quasi-label conditions regarding $\E$-formulas are fulfilled in each world $w$.\footnote{The local conditions and conditions regarding $\A$-formulas have already been fulfilled in $\calM$.} This is clear for worlds of level $\ell(w) > 0$ by the construction of appended paths. So it remains to prove that all $\E$-formulas labeled in worlds $w$ with $\ell(w) = 0$ are satisfied; in fact we show that $\size{cl(w)} \leq \ell(w)$ for all $w$, so for $\ell(w) = 0$ all quasi-label conditions are vacuously satisfied in $w$.

\smallskip

The proof of $\size{cl(w)} \leq \ell(w)$ is by induction on the distance of $w$ from the root $r$. By definition, $\size{cl(r)} \leq \size{cl(\varphi)} = \ell(r)$. Every world $u \neq r$ is of the form $u = \pi[j]$, $j\geq 1$, for a witness path $\pi = f(w,\E\gamma)$. Also, $\ell(w) = \ell(u) + 1$. By construction, $cl(u) \subseteq cl(w)$. We show $cl(u) \neq cl(w)$, so consequently $\size{cl(u)} \leq (\size{cl(w)} - 1) \leq (\ell(w) - 1) = \ell(u)$.

\smallskip

For the sake of contradiction, suppose $cl(u) = cl(w)$.
Let $<$ be a partial ordering on formulas such that $\psi_1 < \psi_2$ iff $cl(\psi_1) \subsetneq cl(\psi_2)$.
Since $cl(u) = cl(w)$,  $\E\gamma \in cl(u)$ for $\gamma = \X\delta$ or $\gamma = \G\delta$. Let $\gamma' \in cl(u)$ be $<$-maximal such that $\E\gamma \in cl(\gamma')$. By the construction of $\pi$, either
\begin{enumerate}
    \item $\gamma' \in cl(\delta)$,
    \item $\gamma' \in cl(\alpha)$ for some $\AX\alpha$ with $w \in L'(\AX\alpha)$,
    \item $\gamma' \in cl(\AF\alpha)$ for some $\AF\alpha$ with $w \in L'(\AF\alpha)\setminus L'(\alpha)$, or
\end{enumerate}
(1) is impossible as $\delta$ is already a proper subformula of $\E\gamma$. In case (2), $\AX \alpha \in cl(w)$ and consequently $\AX\alpha \in cl(u)$. But $\gamma'$ is a proper subformula of $\AX\alpha$, contradicting the $<$-maximality of $\gamma'$ in $cl(u)$. In the final case (3), $w \notin L'(\alpha)$. By the quasi-label condition (Q7), $w \notin L'(\E\gamma)$ despite $f(w,\E\gamma)$ being defined, contradiction.
\end{proof}

\newcommand{\finit}{\ensuremath\calF_\text{root}}
\newcommand{\xinit}{\ensuremath\calX_\text{root}}

\begin{algorithm}\label{alg:alg1}
\caption{NPSPACE algorithm for $\SAT(\B(C, \{ \AF, \AX \}))$}
\DontPrintSemicolon
\SetKwInOut{Input}{Input}\SetKwInOut{Output}{Output}

\Input{$\varphi \in \B(C, \{ \AF, \AX \})$}
\Output{Is $\varphi$ satisfiable?}

\vspace{10pt}

\tcc{$\calG$: $\EG$ formula to satisfy}
\tcc{$\finit$: unfulfilled $\F$-formulas (eventualities) at the root of the balloon}
\tcc{$\xinit$: unfulfilled $\X$ formulas at the root of the balloon}
\tcc{$\ell$: remaining depth counter}
\Procedure{guesspath$(\calG, \finit, \xinit, d)$}{
\lIf{$\ell = 0$}{\Return{\textbf{false}}}
guess $t \in \{ 1, \ldots, 2^{m\cdot \size{\varphi}} \}$ \tcc*{world
where the cycle is closed}
$\calX \dfn \xinit$; $\calF \dfn \finit$; $\calF^* \dfn \emptyset$; $L^* \dfn \emptyset$\\
\For{$i \dfn 0$ to $2^{m\cdot \size{\varphi}}$}
{
guess a set $L \subseteq cl(\varphi)$\\
\lIf{$L$ violates a local quasi-label
condition}{\Return{\textbf{false}}}
\lIf{$\calX \not\subseteq L$}{\Return{\textbf{false}}}
\lIf{$\calG \not\subseteq L$}{\Return{\textbf{false}}}
$\calX \dfn \Set{ \psi | \AX \psi \in L}$\\
\ForEach{$\AF \psi \in L$}{add $\AF \psi$ to $\calF$}
\ForEach{$\psi \in L$}{remove $\AF \psi$ from $\calF$ and $\calF^*$}
\If{$i = t$}{
$L^* \dfn L$ \tcc*{must be equal when closing the cycle}
$\calF^* \dfn \calF$ \tcc*{must be fulfilled before closing the cycle}
}
\ForEach{$\EG \gamma \in L$}
{\lIf{\textbf{not} guesspath$(\{\gamma\}, \calF, \calX,
  \ell-1)$}{\Return{\textbf{false}}}}
  \ForEach{$\EX \xi \in L$}
{\lIf{\textbf{not} guesspath$(\emptyset, \calF, \calX \cup \{ \xi
\}, \ell-1)$}{\Return{\textbf{false}}}}
\lIf{$i\geq t$ and $\calF^* = \emptyset$ and $L =
L^*$}{\Return{\textbf{true}}} }
\Return{\textbf{false}} \tcc*{could not fulfill all eventualities} }
\Return{guesspath$(\emptyset, \emptyset, \{ \varphi \},
m \cdot \size{\varphi})$}
\end{algorithm}

Such a "balloon model" can be constructed in a top-down depth-first search manner to check the satisfiability of $\B(C,\{\AF, \AX\})$ formulas in non-deterministic polynomial space.
A single balloon path is determined by the index of the world where the "back edge" points to, \ie, where the cycle is closed, and then by consecutively guessing the labeled formulas in each world on this path. If an $\E$-formula occurs, then the algorithm recursively guesses witness branches.

This method works in polynomial space, since visited worlds of a branch, unless incident to the back edge, can be "forgotten" immediately.
The correctness of this approach relies on the existence of a normal choice function for witness paths: By injectivity and the disjointness of different witness paths, the algorithm can branch into new sub-balloons independently for each $\E$-formula. Also, the witness paths are not allowed to visit worlds with more than one predecessor, unless it is their root, or the target of their back edge. This allows to track all quasi-label conditions that can affect the worlds on the path, since at any time the algorithm knows the quasi-labels of possible predecessors.

\begin{theorem}\label{thm:afax-in-pspace}
If $T \subseteq \{ \AF, \AX \}$ and $C$ is a base, then $\SAT(\B(C, T)) \in \PSPACE$.
\end{theorem}
\begin{proof}
As $\NPSPACE = \PSPACE$, we consider Algorithm~\ref{alg:alg1} which runs in non-deterministic polynomial space. The previous lemma shows the correctness; the algorithm guesses a balloon quasi-model with a normal choice function by traversing its balloon paths on-the-fly and recursively descending into deeper balloon levels as necessary.

There is an $m \in \N$ such that $\varphi$ is
satisfiable if and only if it has balloon quasi-model with balloon length
$2^{m\cdot \size{\varphi}}$ and level $m\cdot \size{\varphi}$.
A guessed back edge is represented by a pointer $t$ of linear length.
The required space to remember the constantly many sets of labeled formulas is again linear.

Finally the recursion depth is only linear as well, as the depth
of recursion corresponds to the level of the balloon path, so ultimately the
overall space requirement is quadratic.
\end{proof}

\subsection{Hard fragments}\label{sec:hard}

The common proof of the $\EXP$-hardness of the satisfiability problem of CTL is an adaptation of a similar result for PDL by Fischer and Ladner \cite{fischer_propositional_1979}. They use a generic reduction from $\APSPACE$, as $\APSPACE = \EXP$ \cite{alternation}.

$\APSPACE$ (alternating polynomial space) is the class of sets decided
by \emph{alternating polynomial space-bounded single-tape Turing
machines (pspace-ATMs)}. In the following, we show that such machines can be simulated with a wide range of CTL operators, namely $\AU$, $\A\RLS$, and also $\AG$ if combined with $\AX$ or $\AF$.

\medskip

An \emph{alternating Turing machine} is a tuple $M = (Q_\exists,
Q_\forall, \Sigma, \Gamma, \delta, q_0,
\Box, q_\text{acc}, q_\text{rej})$, where $Q_\exists$, $Q_\forall$
are disjoint sets of \emph{existentially} resp.\ \emph{universally branching}
states, $Q \dfn Q_\exists \cup Q_\forall$ is the set of all states,
$q_\text{acc}, q_\text{rej} \in Q$ are the \emph{accepting} resp.\
\emph{rejecting} state, $q_0 \in Q$ is the initial state,
$\Sigma$ and $\Gamma \supsetneq \Sigma$ are the input
and tape alphabet, $\Box \in \Gamma \setminus \Sigma$
is the \emph{blank symbol} and $\delta \colon Q \times \Gamma \to
\powerset{ Q \times \Gamma \times X}$ is the \emph{transition function},
where $X = \{ -1, 0, 1\}$ and $\delta(q,a)$ is a finite set for all $q
\in Q, a \in \Gamma$.

A \emph{configuration} is a tuple $(q, i, t)$, where $q \in Q$ is the
current state, $i \in \N$ is the current \emph{head position}
and $t \in \Gamma^*$ the current \emph{tape content}, \ie, $t$ is a finite word $(c_1,\ldots,c_k)$ consisting of symbols of $\Gamma$, and $i \in [k]$.
Write $\delta(q,i,t)$ for the set of all configurations resulting from
applying a transition of $\delta(q,t_i)$.
Provided that $q \neq q_{\text{rej}}$, a configuration \emph{accepts} if $q = q_{\text{acc}}$; or if $\delta(q,i,t)$ contains
at least one configuration that accepts and $q \in Q_\exists$; or if
it contains only accepting configurations and $q\in Q_\forall$.
$M$ accepts an input $x \in \Sigma^*$ if the initial
configuration $(q_0, 1, x)$ accepts. $M$ runs
in \emph{polynomial space} if there is a polynomial $g$ such that on each input $x$ the head position $i$ is always in $[g(\size{x})]$ (we can assume that $M$ does not leave the
input to the left of position 1).

\begin{theorem}\label{thm:expc}
$\SAT(\B_2(T))$ is $\EXP$-hard if $\AU
\in T$, $\A\RLS \in T$, $\{\AG, \AX \} \subseteq T$ or $\{ \AG, \AF \} \subseteq T$.
\end{theorem}
\begin{proof}
Let $A \in \EXP$. As $\EXP = \APSPACE$ \cite{alternation}, $A$ is decided by a pspace-bounded ATM $M = (Q_\exists,
Q_\forall, \Sigma, \Gamma, \delta, q_0,
\Box, q_\text{acc}, q_\text{rej})$. \Wloss $\delta(q,a)$ is always non-empty, and on all inputs every computation path eventually assumes the state $q_\text{acc}$ or
$q_\text{rej}$. That such an $M$ can be chosen is proved similar to \citeref{alternation}{Thm. 2.6}.

We reduce $A$ to $\SAT(\B_2(T))$ via $M$.

\medskip

\textbf{Case 1:  $\AG,\AX$}

\smallskip

The following CTL formula $\varphi \in \B_2(\{\AG, \AX\})$ is
satisfiable if and only if $M$ accepts $x$.
$\varphi$ is constructible in space that is logarithmic in $\size{x}$. Let
$I \dfn [g(\size{x})]$.
\begin{align*}
\varphi \dfn \; &\varphi_\text{init} \sland \AG \varphi_\text{conf}
\sland \AG \varphi_\delta \\
\phiinit \dfn \; &s_{q_0} \sland p_1 \sland \bigand_{\mathclap{1 \leq i \leq
\size{x}}} t_{i,x_i} \sland \bigand_{\mathclap{\substack{i \in I \\ i >
\size{x}}}} t_{i,\Box}\\
\phiconf \dfn \; &\bigor_{q\in Q} \left(s_q \sland
\bigand_{\mathclap{\quad \; q' \in Q\setminus \{q\}}} \neg s_{q'}\right)
\sland
\bigor_{i \in I} \left(p_i \sland \bigand_{\mathclap{\quad \; j \in
I\setminus\{i\}}} \neg p_{j} \right) \; \sland
 \bigand_{i \in I} \bigor_{a \in \Gamma} \left(t_{i,a}
\sland \bigand_{\mathclap{\quad \; a' \in \Gamma \setminus \{ a\}}} \neg t_{i,a'} \right)\\
\varphi_\delta \dfn \; & s_{q_\text{acc}} \slor \Bigg(
\neg s_{q_\text{rej}} \sland
\bigand_{\mathclap{\substack{q \in Q_\exists\\ a \in \Gamma\\ i \in I}}}
\bigg((s_q \sland p_i \sland t_{i,a})
  \rightarrow \bigor_{\mathclap{\substack{(q',a',X)\\ \in
\delta(q,a)}}}
\phinexti{q'}{i}{i+X}{a'}\bigg) \\
& \phantom{s_{q_\text{acc}} \slor \neg s_{q_\text{rej}}} \land \;\bigand_{\mathclap{\substack{q \in Q_\forall\\ a\in \Gamma\\ i \in I}}}
\bigg((s_q \sland p_i \sland t_{i,a}) \rightarrow
\bigand_{\mathclap{\substack{(q',a',X)\\ \in \delta(q,a)}}}
\phinexti{q'}{i}{i+X}{a'}\bigg)\Bigg)\\
\phinexti{q'}{i}{i'}{a'} \dfn \; &\EX (s_{q'} \sland p_{i'}
\sland t_{i,a'}) \; \sland \; \bigand_{\mathclap{\substack{j \in I
\\
j \neq i \\ a \in \Gamma}}}
\big((t_{j,a} \rightarrow \AX t_{j,a}) \sland (\neg t_{j,a}
\rightarrow \AX \neg t_{j,a})\big)
\end{align*}

$\phiinit$ fixes the root of models of $\varphi$ to
simulate the initial configuration of $M$ on $x$.
$\AG\phiconf$ forces every reachable world to assume exactly one
configuration of $M$. $\AG\varphi_\delta$ requires the existence of
successor configurations resulting from $\delta$-transitions (and is
falsified if $q$ is the rejecting state), and finally
$\phinexti{q'}{i}{i'}{a'}$ fixes all tape symbols at positions where the
head currently does not write.
Now it holds that $\varphi$ is satisfiable if and only if the initial
configuration of $M$ is accepting.
At this point it is crucial that all computation paths of $M$
eventually accept or reject. This allows a correct reduction even without "eventuality" operators.

\medskip

\textbf{Case 2:  $\AG,\AF$}

\smallskip

Without $\AX$, it is harder to express that the worlds quantified inside $\phinexti{q'}{i}{i'}{a'}$ coincide, \ie, that the world representing the successor configuration is exactly the world where all non-overwritten symbols stay the same.
Obviously it will not work to just replace $\AX,\EX$ with $\AF,\EF$.
As a solution, we do not quantify successors, but whole infinite paths. Each such path then assumes a
single reachable configuration and must eventually continue the
computation. Change several formulas as follows.

\begin{align*}
\phinexti{q'}{i}{i'}{a'} \dfn \; &
\EG\Big[\bigand_{\mathclap{\substack{j \in I\\ j \neq i
\\ a \in \Gamma}}} \underline{\phikeep{j}{a}}\, \sland \Big(\underline{(q',i',a')}
\rightarrow (s_{q'} \sland p_{i'} \sland t_{i,a'})\Big)\Big] \land
\AF \underline{(q',i',a')}\\
\phikeep{j}{a} \dfn \; &
\big(\AF(b \sland t_{j,a}) \rightarrow (t_{j,a} \sland
\AF(\neg b \sland t_{j,a}))\big) \; \sland\\
&\big(\AF(\neg b \sland t_{j,a}) \rightarrow (t_{j,a} \sland \AF(b \sland
t_{j,a}))\big)\\
\varphi \dfn \; &\varphi_\text{init} \sland \AG \varphi_\text{conf}
\sland \AG \varphi_\delta \sland \AG
\bigand_{\mathclap{\substack{j \in I\\ a \in \Gamma}}}
\;\left(\,\underline{\phikeep{j}{a}} \leftrightarrow \phikeep{j}{a}\right)
\end{align*}

Here, the underlined expressions are atomic propositions. The
formula $\phikeep{j}{a}$ does not directly occur to retain a
low temporal depth.

We proceed by showing that $\pi \vDash t_{j,a} \rightarrow \G t_{j,a}$ for any path $\pi$  that fulfills $\phikeep{j}{a}$.
The following proof works by induction on the length $\ell$ of a prefix of $\pi$. The case $\ell = 1$ is clear.
Let $\ell > 1$. $\pi[\ell]$ satisfies either $b$ or $\neg b$, assume \wloss{} $\pi[\ell] \vDash b \land t_{j,a}$.
Then $\pi[\ell] \vDash b \land \AF(\neg b \sland t_{j,a})$. By definition of $\AF$, the
world $\pi[\ell + 1]$ must fulfill $\AF[\neg b \sland t_{j,a}]$ and consequently $t_{j,a}$ due to $\phikeep{j}{a}$.

The modified $\phinexti{q'}{i}{i'}{a'}$ eventually enforces a
reachable world $w$ to assume a successor configuration. All tape
symbols at position $j \neq i$ remain unchanged. Then the computation
continues from $w$ on fresh paths starting at $w$ (where then the
tape symbol at the new positions $i'$ can change and all others are fixed).

\medskip

\textbf{Case 3:  $\AU$}

\smallskip

We further modify the approach in the previous case.
To replace $\AG$, we use the fact that the computation tree has only to
be verified to be legal until a point where $q_\text{acc}$ or $q_\text{rej}$ is reached.
We introduce a new proposition $h$
(\emph{halted}) and replace every $\AG\psi$ by $\A[\psi \U h]$. Replace $\AF\underline{(q',i',a')}$ by
$\A[\neg h \U \neg h \sland \underline{(q',i',a')}]$, every other
$\AF\psi$ by $\A[\top \U \psi]$, and $\EG\psi$ by $\neg \A[\top \U \neg \psi]$.
This ensures that $\neg h$ holds as
long as the computation is continued, but also allows that the paths
not usable for further computation (as they fixed all tape symbols but one)
can label $h$ after $\underline{(q',i',a')}$.

\medskip

\textbf{Case 4:  $\A\RLS$}

\smallskip

As $\AG \psi \equiv \A[\bot\RLS\psi]$, we extend the $\AG,\AX$
case and only modify $\phinexti{q'}{i}{i'}{a'}$:
\begin{align*}
\phinexti{q'}{i}{i'}{a'} \dfn \; & \bigand_{\substack{q \in Q\\ a \in \Gamma}}
(s_q \sland p_i \sland t_{i,a}) \rightarrow
 \E[(s_q \sland p_i \sland t_{i,a}) \U (s_{q'} \sland p_{i'} \sland
 t_{i,a'})]\\
& \qquad\sland \; \bigand_{\mathclap{\substack{j \in I\\ i \neq
j \\ c \in \Gamma}}} \left(t_{j,c}
\rightarrow
\A[\neg(s_q \sland p_i \sland t_{i,a}) \RLS
t_{j,c}]\right)
\end{align*}
The formula $\phinexti{q'}{i}{i'}{a'}$ requires a reachable world where
eventually $s_{q'} \sland p_{i'} \sland t_{i,a'}$ holds.
The $\A\RLS$ subformulas state for all $j \neq i$ that
$\neg(s_q \sland p_i \sland t_{i,a})$ \emph{releases} $t_{j,c}$, \ie, the
earliest world where $t_{j,c}$ no longer has to hold is exactly the
world \emph{after} the one where the $\EU$ is fulfilled (\wloss one of $q$, $i$ or
$t_i$ changes in the transition). This again fixes the tape symbols that are not changed in the transition.
\end{proof}

For a CTL formula to simulate the computation of a polynomially space bounded machine, it is necessary that it can enforce exponentially long paths. This lower bound will be shown for the four fragments from the previous theorem.
The cases where $T$ contains $\AU$ or $\AF$ follow from Corollary~\ref{cor:model-size-af}, as $\AF \psi \equiv \A[\top \U \psi]$. It remains to consider $\A\RLS$ and $\{\AG,\AX\}$.

The fragment $\B(\{\AG, \AX\})$ is almost similar to the modal logic $\mathsf{KD}$ enriched with the \emph{universal modality} $\boxast$. The main difference is that $\boxast \varphi$ usually means that $\varphi$ holds in \emph{all} worlds of a model, but $\AG$ only refers to reachable worlds.
Nevertheless, the modal logic $\KD + \boxast$ can enforce a model of depth $2^n$ with a formula of size
$\bigO{n^2}$ via the construction of a binary counter \cite{FMT_2007}, using $\boxast$ only in the root. This approach is again translated to also work with $\AU,\A\RLS$ and $\{\AG,\AF\}$.

\begin{theorem}\label{thm:model-size-au-ar}
If $\AU \in T$, $\A\RLS \in T$, $\{\AG, \AX\} \subseteq T$ or $\{\AG,\AF\} \subseteq T$, then
$\B_2(T)$ has extent lower bound $2^{\bigOmega{n}}$.
\end{theorem}
\begin{proof}
We simulate the approach of Grädel et al.\ \cite{FMT_2007}, using $\AG$ and $\AX$, and further optimize it with a few extra propositions to obtain a formula that does the same but has only linear length. The formula is defined as follows.

\begin{align*}
\alpha \dfn\;& \left(p_0 \leftrightarrow \text{carry}_{\leq 0}\right)
\sland \bigand_{i = 1}^{n} \big(p_i \sland \text{carry}_{\leq i-1}
\leftrightarrow \text{carry}_{\leq i} \big) \sland\\
& \big(\text{reset}_{\leq 0} \rightarrow \neg p_0\big) \sland
\bigand_{i = 1}^{n} \left(\text{reset}_{\leq i} \rightarrow \neg
p_i \sland \text{reset}_{\leq i - 1}\right) \sland\\
& \phantom{\big(\text{reset}_{\leq 0} \rightarrow \neg p_0\big) \sland \; }
\bigand_{i = 1}^{n} \left(\text{store}_{\geq i-1} \rightarrow
\text{store}_{i-1} \sland \text{store}_{\geq i}\right)\text{,}\\
\beta \dfn \;& \bigand_{i = 1}^{n}\left(\text{store}_{i}
\rightarrow (p_i \rightarrow \AX p_i) \sland (\neg p_i \rightarrow
\AX \neg p_i)\right)\\
\gamma \dfn \;& \bigand_{i = 1}^{n}\left((\text{carry}_{\leq i - 1}
\sland \neg p_i) \rightarrow \AX\big( p_i \land \text{reset}_{\leq
i - 1}\big) \sland \text{store}_{\geq i+1} \right)\\
& \qquad \sland (\neg p_0 \rightarrow \AX p_0 \sland \text{store}_{\geq 1})
\\
\varphi \dfn \;&\AG(\alpha \sland \beta \sland \gamma) \sland
\bigand_{i = 0}^{n} \neg p_i
\end{align*}

The idea is the same as in \cite{FMT_2007}: The propositions $p_i$
form a binary counter of length $n$ that assumes the values $0 \ldots
2^n-1$ in this order. The value 0 is assumed in the root of the model.
If the propositions in a world $w$ form the counter value $k$, they
are forced to form $k+1$ in every successor world of $w$. This is expressed in the subformula $\gamma$: Search for the least significant bit with
value 0 that has only 1s to the right.
Force it to flip in the next
world, but also flip all the bits to the right to 0. The higher
significant bits may not change between $w$ and its successor, which is ensured by
$\beta$ and $\gamma$.

The use of the formula $\alpha$ improves the formula length from
$\bigO{n^2}$ to $\bigO{n}$. The new propositions work as follows: carry$_{\leq i}$ is true if and only if all bits at position $\leq i$ were set to one and the incrementation causes a carry bit at position greater than $i$. It depends only on $p_i$ and carry$_{\leq i-1}$, which avoids repeated inner conjunctions like $\bigwedge_{j=0}^{j} p_j$ to determine whether there is a carry at position $i$.
Similarly, to set all positions $\leq i$ back to zero, reset$_{\leq i}$ is used to avoid $\bigwedge_{j=0}^{i}\neg p_j$; and to keep all positions $\geq i$ unchanged, store$_{\geq i}$ avoids the formula $\bigwedge_{j=0}^i p_j \imp \AX p_j$.

When using $\A\RLS$, we can define $\AG$ and $\EU$ but not $\AX$,
so more work is required. In particular, we have to distinguish two
cases: Whether the counter value changes from even to odd, \ie, the
only changing bit is $p_0$ and it changes from zero to one, or it
changes from odd to even, \ie, $p_0$ flips from one to zero.

In $\gamma$, replace $\AX\big( p_i \sland \text{reset}_{\leq i -
1}\big)$ by $\E\big[p_0 \U (p_i \sland \text{reset}_{\leq i -
1})\big]$ (the odd-to-even case) and $\AX p_0$ by $\E[\neg p_0 \U p_0]$ (the
even-to-odd case).
This formula flips the correct bit $p_i$ from zero to one as well as
lesser significant bits from one to zero in some reachable world, which
is however not necessarily a direct successor.
To retain the values of more significant bits until this world is
actually reached, change $\beta$ to:
\begin{align*}
\beta \dfn \; \bigand_{i = 1}^{n} \bigg(\text{store}_{i}
\rightarrow
&\phantom{\; \sland \, }\Big(\phantom{\neg} p_0 \rightarrow (\phantom{\neg }p_i \rightarrow
\A[\neg p_0 \RLS \phantom{\neg }p_i])\\
& \phantom{(\neg p_0 \rightarrow}
\; \sland \, (\neg p_i \rightarrow \A[\neg
p_0 \RLS \neg p_i])\Big)\\
&\sland \, \Big(\neg p_0 \rightarrow (\phantom{\neg }p_i \rightarrow \A[\phantom{\neg } p_0
\RLS \phantom{\neg }p_i])\\
& \phantom{(\neg p_0 \rightarrow}
\; \sland \, (\neg p_i \rightarrow \A[\phantom{\neg }
p_0 \RLS \neg p_i])\Big)\bigg)
\end{align*}

The above formula preserves the state of the corresponding $p_i$ until the first change of $p_0$.
However, the $\EU$-subformulas of $\gamma$ are chosen to maintain the state of $p_0$ until the actual point of fulfillment.
Accordingly, all bits of higher significance are preserved until this world, and altogether
there is a simple path that assumes all the counter values $0\ldots
2^n-1$ at least once.
\end{proof}
 
\clearpage
\section{Flat CTL}\label{sec:flat-ctl}

The previous section has established lower bounds, in complexity and model size, for temporal depth of at least two.
This section, on the other hand, investigates the corresponding
fragments of \emph{flat} CTL, \ie, with temporal depth at most one.
In contrast to the fragments with operator nesting permitted, all flat cases have the polynomial model property.

We start with using only the operators $\AX$ and $\AG$.

\begin{theorem}\label{thm:minimal-model-size-axag}
Let $C$ be a base. If $\emptyset \subsetneq T \subseteq \{ \AX, \AG \}$, then $\B_1(C,T)$ has optimal model size $\bigO{n}$ and extent $\leq \size{T}$.
\end{theorem}
\begin{proof}
Let $\varphi \in \B_1(C,T)$ be satisfiable.
$\varphi$ is logically implied by a satisfiable formula of the form
\[
\varphi' = \bigwedge_{i=1}^m \E\psi_i \land \bigwedge_{i=1}^k \A \xi_i\text{,}
\]
where $\size{\varphi'} \in \bigO{\size{\varphi}}$. (Since $\varphi$ is a Boolean combination of CTL formulas, think of $\varphi'$ as a "satisfying assignment".)

It is clear that in the cases $T = \{ \AX \}$ and $T = \{ \AG \}$, all $\E$-subformulas $\varphi$ can be fulfilled in distinct successors of the root. The extent is then $1$.
If however $T = \{ \AG, \AX \}$, then an $\AX$-subformula can prevent an $\EF$-formula from being fulfilled in an immediate successor.
Nevertheless, the minimal extent is then at most $2$.
Clearly, in any model of $\varphi'$ with extent $\leq 2$, all but $m$ worlds of distance 1 and all but $m$ worlds of distance $2$ can be deleted to reach the size upper bound.
\end{proof}

\begin{theorem}\label{thm:lower-model-size-axag}
Let $\emptyset \subsetneq T \subseteq \{ \AX, \AG \}$. Then $\B_1(T)$ has optimal model size $\bigOmega{n}$ and extent $\geq \size{T}$.
\end{theorem}
\begin{proof}
Consider the formula family $(\varphi_m)_{m \in \N}$ defined by
\begin{align*}
\varphi_m \dfn \; &\AX \sigma(p_1,\ldots,p_m) \land \bigand_{i = 1}^{m} \EX p_i\\
\sigma(p_1,\ldots,p_m) \dfn &\bigand_{i = 1}^m (p_i \imp (\hat{p}_{<i}\land \hat{p}_{>i})) \;\land\; \bigwedge_{i=2}^{m} (\hat{p}_{<i} \imp (\hat{p}_{<i-1}\land \neg p_{i-1})) \;\land\\
&\phantom{\bigand_{i = 1}^m (p_i \imp (\hat{p}_{<i}\land p_{>i})) \;\land\;}\;\bigwedge_{i=1}^{m-1}( \hat{p}_{>i} \imp (\hat{p}_{>i+1}\land \neg p_{i+1}))\text{.}
\end{align*}
The idea is that every world satisfying $\sigma(p_1,\ldots,p_m)$ can have at most one of $p_1,\ldots,p_m$ true.
For this, we implement "carry propositions" $\hat{p}_{<i}$ and $\hat{p}_{>i}$ as in Theorem~\ref{thm:model-size-au-ar}.
Then $\varphi_m$ is satisfiable and has length $\bigO{m}$, but any model of $\varphi_m$ has at least $m$ worlds.
For $\AG$/$\EF$ instead of $\AX$/$\EX$ the formula works analogously.
The minimal extent is $1$ for the formulas $p \land \EX \neg p$ and $p \land \EF \neg p$, and $2$ for $(p \land q) \land \AX(p \land \neg q) \land \EF (\neg p \land q)$.
\end{proof}

In the case where all CTL operators are available, both the size and the extent bounds increase by a factor of $n$:

\begin{theorem}\label{thm:flat-ctl-upper-bounds}
Let $C$ be a base and $T \subseteq \TL$. Then $\B_1(C,T)$ has optimal model size  $\bigO{n^2}$ and extent $\bigO{n}$.
\end{theorem}
\begin{proof}
Let $\varphi \in \B_1(C,T)$ be satisfiable. \Wloss $T \subseteq \{\AX,\AU,\A\RLS\}$.
As in the proof of Theorem~\ref{thm:minimal-model-size-axag}, $\varphi$ is logically implied by a satisfiable formula of the form
\[
\varphi' = \bigwedge_{i=1}^m \E \psi_i \land \bigwedge_{i=1}^k \A \xi_i\text{,}
\]
where $\size{\varphi'}\in \bigO{\size{\varphi}}$.
$\varphi'$ has a model $\calK$ that consists of a root $w_0$ and $m$ otherwise disjoint branches $\pi_1,\ldots,\pi_m$ such that $\pi_i \vDash \psi_i$.
\Wloss these branches end in self-loops.
In the following we show that every branch can be shrunken down to at most $\bigO{k}$ worlds. This then proves the theorem.

We mark worlds on $\pi_i$ as follows. First, mark $\pi_i[0]$ and $\pi_i[1]$.
For every $\xi_j = \vartheta \U \vartheta'$, mark the first worlds where $\vartheta'$ holds. For $\xi_j = \vartheta'\RLS \vartheta$, proceed similarly, provided that such a world exists.
Likewise, mark the world that fulfills $\psi_i$, if such a world exists.
Then clearly $\pi_i$ can be replaced by a subpath consisting of all $\leq (k + 2)$ marked worlds, arranged in the same order as before, without violating $\E\psi_i$ or any $\A\xi_j$.
\end{proof}

For the corresponding lower bound, we identify several CTL operators that have the capability to enforce a model consisting of $n$ disjoint paths of length $n$.

\begin{theorem}\label{thm:model-size-flat-ur-lower}
Let $T$ contain $\AU$, $\A\RLS$ or $\{\AG,\AF\}$. Then $\B_1(T)$ has optimal model size  $\bigOmega{n^2}$ and extent  $\bigOmega{n}$.
\end{theorem}
\begin{proof}
Let $\AG,\AF \in T$.
Let the formula $\sigma_p$ state that at most one of $p_1,\ldots,p_m$ is true, and let the formula $\sigma_q$ state that at most one of $q_1,\ldots,q_m$ is true (independently of $p_1,\ldots,p_m$).
Such formulas can be constructed as in the proof of Theorem~\ref{thm:lower-model-size-axag}.
Then let
\[
\varphi_m \dfn \AG (\sigma_p \land \sigma_q) \land \bigwedge_{i=1}^m \EG (r \lor p_i) \land \bigwedge_{j=1}^{m} \AF (q_i \land \neg r)\text{.}
\]
$\varphi_m$ has length $\bigO{m}$ and is satisfiable.
But any model of $\varphi_m$ must satisfy $q_1,\ldots,q_m$ in $m$ distinct worlds on every path.
Moreover, paths $\pi_1,\ldots,\pi_m$ must exist with $r\lor p_i$ holding globally on $\pi_i$.
These paths are disjoint in the fulfillment points of $\AF (q_1\land\neg r), \ldots, \AF (q_m \land \neg r)$.
As a result, any model has size at least $m^2$ and extent $m$.

For the case $\AU \in T$, change the above formula to
\[
\varphi_m \dfn \A[\sigma_p \U (\sigma_p \land q_{m})] \land \bigwedge_{i=1}^{m} \E[q_m \RLS (p_i\lor r)] \land\bigand_{i=1}^{m-1} \A[\neg q_{i+1} \U (\neg r \land q_i)] \text{.}
\]
Due to the first conjunction, a world with $q_m$ is reached on any path, with $\sigma_p$ being true until that point.
However, by the last conjunction, on every path the propositions $q_1,\ldots,q_{m-1}$ must appear before $q_m$ exactly in this order.
Due to the middle conjunction, there are at least $m$ such paths, and again any model has at least $m^2$ worlds and extent $m$.

Finally, for $\A\RLS \in T$, the formula
\[
\varphi_m \dfn \A[\bot \RLS \sigma_p] \land \bigwedge_{i=1}^m \E[(p_i \lor r)\U q_m] \land \bigwedge_{i=1}^m \A[(q_i \land \neg r) \RLS \neg q_{i+1}]
\]
works analogously.
Due to the middle part, the last conjunction of $\A\RLS$s cannot be fulfilled by simply having $\neg q_1,\ldots,\neg q_m$ true indefinitely.
Instead, $q_1,\ldots,q_m$ have to be fulfilled one after another on every path, and a similarly structured model as in the other cases is enforced.
\end{proof}

If the CTL operators are restricted to $\{\AF\}$ or $\{\AF,\AX\}$, then the above construction does not work due to the "mixed quantifier" nature of $\AF$ and $\EG$.
Instead, a formula that enforces $n$ worlds in a model is already of length $n \log n$.

To express such a model size in terms of the length of the corresponding formula, we require a function $w$ such that $w^{-1}(n) = n \log n$.
A function satisfying this equation, at least asymptotically, is $w \,\colon \mathbb{R}_+ \to \mathbb{R}_+$ with $w(x) \dfn \frac{x}{W(x)}$, where $W(x)$ is the \emph{Lambert W function} \cite{corless_lambert_1996},
the inverse function of $W^{-1}(x) \dfn x e^x$.

\begin{proposition}
For all $x \in \mathbb{R}_+$, $w(x \ln x) = x$, that is, $w^{-1}(x) = x \ln x$.
\end{proposition}

\begin{theorem}\label{thm:min-size-af}
Let $\AF \in T$. Then $\B_1(T)$ has optimal model size $\bigOmega{w(n)^2}$ and extent $\bigOmega{n}$.
\end{theorem}
\begin{proof}
Consider the formula family $(\varphi_m)_{m \in \N}$ defined by
\[
\varphi_m \dfn \bigand_{i = 0}^{m-1} \AF (\vec{c}(i) \land \neg r) \land \bigand_{i = 0}^{m-1} \EG (r \slor \vec{d}(i))\text{,}
\]
where $\vec{c}(i)$ and $\vec{d}(i)$ are conjunctions of
$\ceil{\log m}$ literals representing the value $i$ as a binary vector, similarly as in Theorem~\ref{theorem:model-size-ax-lower}.

$\varphi_m$ is satisfiable, but any model of it contains $m^2$ worlds as in Theorem~\ref{thm:model-size-flat-ur-lower}.
For a constant $k$, we can set $n \dfn k \cdot m \ln m$ and obtain an infinite family of formulas of size $\leq n$ and models with size at least $w(\frac{n}{k})^2$. Since $w(\frac{n}{k})\geq \frac{1}{k} w(n)$ for large enough $k$, it follows $w(\frac{n}{k})^2 \in \bigOmega{w(n)^2}$.

For the minimal extent $\bigOmega{n}$, the formula $\varphi_m \dfn \EG\sigma_p \land \bigwedge_{i=1}^m \AF p_i$ works similarly as in the proof of Theorem~\ref{thm:model-size-flat-ur-lower}.
\end{proof}

\subsection{Existential Flat CTL}\label{sec:ex-flat-ctl}

In the absence of universal path quantifiers, even smaller models can be found.
Whenever the formulas $\EX\psi_1, \ldots, \EX\psi_n$ are satisfiable, they can be fulfilled in the same model.
Lower bounds for the model size can then only stem from, say, $\psi_i \land \psi_j$ being not satisfiable in a \emph{single} successor.

Formally, $\varphi \in \B_1(C)$ is called \emph{existential} if it is a monotone Boolean combination of propositional formulas and $\E$-preceded CTL formulas.
In this setting, model size lower bounds emerge that depend solely on the number of contradicting subformulas.
Propositional formulas $\psi, \psi'$ are \emph{contradicting} if $\psi$ and $\psi'$ are both satisfiable, but $\psi\land\psi'$ is not.

\smallskip

Our goal is to determine the maximal number of contradicting subformulas that a formula with a given length can exhibit. We reduce this problem to a graph-theoretical problem called \emph{biclique covering}. Recall that a \emph{biclique} $A \times B$ is a complete bipartite graph.\footnote{That is, a graph with vertices $A \cup B$ such that $A$ and $B$ are disjoint, and with the edge $(u,v)$ existing if and only if $u\in A, v\in B$ or $u \in B, v \in A$.}

\begin{definition}
Let $G = (V,E)$ be a graph. A \emph{biclique covering of $G$} is a sequence $(A_i \times B_i)_{i\in [n]}$ of biclique subgraphs of $G$ such that $\bigcup_{i\in [n]}(A_i \times B_i)\cup (B_i \times A_i) = E$. Its \emph{weight} is $\bigwedge_{i\in[n]}|A_i|+|B_i|$.

The \emph{minimal biclique covering weight} of $G$ is the minimal weight of a biclique covering of $G$.
\end{definition}

\begin{proposition}[\citeref{jukna2011extremal}{~p.\!~46}]\label{prop:minimal-cover-weight}
The $n$-vertex clique graph $K_n$ has a minimal biclique covering weight of at least $n \log n$.
\end{proposition}

\newcommand{\psize}[1]{\langle{}#1\rangle{}}

We apply the above result in the following lemmas.
If $\psi$ is a formula, $\psize{\psi}$ denotes the total number of occurrences of propositions in $\psi$. For example, $\psize{p \lor \neg p} = 2$.
Clearly $\psize{\psi} \leq \size{\psi}$.

\begin{lemma}\label{lem:contradict}
Let $C$ be a base, and let $\psi,\ldots,\psi_n \in \B_0(C)$ be pairwise contradicting.
Then $\sum_{i \in [n]}\psize{\psi_i} \geq n \log n$.
\end{lemma}
\begin{proof}
As $\psi_i$ is satisfiable for all $i \in [n]$, it is implied by a satisfiable conjunction $\psi^*_i$ of literals (\ie, over the base $\{\land,\neg\}$), such that $\psize{\psi^*_i}\leq \psize{\psi_i}$.
It follows that $\psi^*_1,\psi^*_2, \ldots, \psi^*_n$ are again pairwise contradicting.
For this reason, proving the lower bound for conjunctions of literals is sufficient.

We use the $n$-vertex clique graph $K_n$ with vertices $\{\psi_i \mid i \in [n]\}$.
The goal is to cover all edges, where a covered edge $(\psi_i, \psi_j)$
in $K_n$ means that $\psi_i$ and $\psi_j$ are contradicting.
To cover an edge, we require the literal $p$
occurring in the conjunction $\psi_i$ and $\neg p$ occurring in $\psi_j$ (or vice versa).

Consider the subset $A_p \subseteq \{\psi_1,\ldots,\psi_n\}$ such that $\psi_i$ implies $p$ for all $\psi_i \in A_p$, and similarly $B_p \subseteq \{\psi_1,\ldots,\psi_n\}$ such that $\psi_i$ implies $\neg p$ for all $\psi_i \in B_p$.
$A_p$ and $B_p$ are disjoint.
$A_p \cup B_p$ does not necessarily contain all vertices of $K_n$;
nevertheless, every $\psi_i\in A_p$ contradicts each $\psi_j\in B_p$.
Consequently, the edges covered due to the proposition $p$ form a biclique $A_p \times B_p$, and $\size{A_p} \cdot \size{B_p}$ edges in $K_n$ are covered.
The weight $\size{A_p} + \size{B_p}$ of the biclique is simultaneously a lower bound for the number of occurrences of the proposition $p$ in $\psi_1,\ldots,\psi_n$, as it must occur in all formulas in $A_p \cup B_p$ to make them contradicting.
Ultimately, if $p_1,\ldots,p_m$ are the propositions occurring in $\psi_1,\ldots,\psi_n$, then $A_{p_1} \times B_{p_1}, \ldots, A_{p_m} \times B_{p_m}$ must form a biclique covering of $K_n$.

However, by Proposition~\ref{prop:minimal-cover-weight}, the minimal biclique covering weight of $K_n$ is $n \log n$.
This is then also the minimal number of occurrences of variables counted over all $\psi_1,\ldots, \psi_n$ (with at least one variable per biclique), which proves the lemma.
\end{proof}

The next lemma is the corresponding upper bound.
It states that formulas that are \emph{not} pairwise contradicting can be "merged".

\begin{lemma}\label{lem:merge-upper-bound}
Let $\Phi = \{ \psi_i \mid i \in [n]\} \subseteq \B_0(C)$ be a set of satisfiable formulas.
Then there is a partition $\Phi_1 \dot\cup \cdots \dot\cup \,\Phi_m = \Phi$ such that
\begin{enumerate}
\item $m \log m \leq \sum\limits_{i\in [n]} \psize{\psi_i}$,\vspace{3pt}
\item $\bigwedge\limits_{\psi \in \Phi_j} \psi$ is satisfiable for every $j \in [m]$.
\end{enumerate}
\end{lemma}
\begin{proof}
Let $m$ be minimal such that $\Phi_1,\ldots,\Phi_m$ is a partition that satisfies (2).
Such an $m$ must exist.
Let $\psi^*_j \dfn \bigwedge_{\psi \in \Phi_j}\psi$ for $j \in [m]$.
The formulas $\psi^*_1,\ldots,\psi^*_m$ are pairwise contradicting, otherwise we could coarsen the partition and $m$ would not be minimal.
By the previous lemma, then $\sum_{i \in [m]}\psize{\psi^*_i} \geq m \log m$.
But as $\Phi_1,\ldots,\Phi_m$ is a partition of $\Phi$, it holds $\sum_{i\in [n]} \psize{\psi_i} = \sum_{i \in [m]} \psize{\psi^*_i}$, and (1) follows.
\end{proof}

\begin{theorem}
Let $C$ be a base. Then existential $\B_1(C)$ has optimal model size $\bigO{w(n)}$ and extent $\leq 1$.
\end{theorem}
\begin{proof}
Let $\varphi \in \B_1(C)$ be existential and satisfied by a model $\calM = (\calK,w)$.
It holds $\varphi = f(\psi_0,\psi_1,\ldots,\psi_m)$ for a monotone Boolean combination $f$, $\psi_0 \in \B_0(C)$, and $\E$-preceded arguments $\psi_1, \ldots, \psi_m$.

Let $I \dfn \{i \in \{0,\ldots,m\} \mid \calM \vDash \psi_i \}$.
Our goal is to transform $\calM$ to a model of $\varphi$ of size $\bigO{w(n)}$.
By the monotonicity of $f$, any transformation of $\calM$ that preserves the truth of $\bigwedge_{i \in I} \psi_i$ will suffice.
Similarly as in Theorem~\ref{thm:flat-ctl-upper-bounds}, assume that every $\psi_i$ be fulfilled on a distinct branch in $\calM$.

In the next step, we aim to simplify all temporal operators to $\EX$. On that account, we define for every $\psi_i$ a "reduct"  $\tau_\calM(\psi_i)$ such that $\tau_\calM(\psi_i)$ entails $\psi_i$ but is still true in $\calM$.
For instance, \wloss every $\psi_i = \EF \xi$ is fulfilled in $w$ or a successor of $w$.
Consequently, we define $\tau_\calM(\psi_i)$ as $\xi$ or $\EX \xi$.
Similarly, for $\psi_i = \E[\xi \U \xi']$ let $\tau_\calM(\psi_i) = \xi'$ or $\tau_\calM(\psi_i) = \xi \land \EX \xi'$.
The cases $\psi_i = \EG\xi$ and $\psi_i = \E[\xi' \RLS \xi]$ already imply $\calM \vDash \xi$. As $\calM$ can be assumed reflexive, in both cases let $\tau_\calM(\psi_i) \dfn \xi$.

Then by definition, the conjunction $\bigwedge_{i \in I} \tau_\calM(\psi_i)$ entails $\varphi$ on reflexive models and is satisfiable.
It can be written as\[
\varphi' \dfn \bigwedge_{i=1}^k \xi_i \land \bigwedge_{i=k+1}^{\ell} \EX \xi_i\text{,}
\]
where $\size{\varphi'} \in \bigO{\size{\varphi}}$ and $\xi_i \in \B_0(C)$ for $i \in [\ell]$.
By Lemma~\ref{lem:merge-upper-bound}, we can satisfy $\xi_1,\ldots,\xi_k$ in $w(\size{\varphi'}) \in \bigO{w(\size{\varphi})}$ worlds.
\end{proof}

\begin{theorem}
Existential $\B_1$ has optimal model size $\bigTheta{w(n)}$ and extent $\geq 1$.
\end{theorem}
\begin{proof}
Consider $(\varphi_m)_{m \in \N}$ defined by $\varphi_m \dfn \bigand_{i = 1}^{m+1} \EF \vec{c}(i)$, with $\vec{c}(i)$ representing $i$ as binary vector as in Theorem~\ref{thm:min-size-af}.
$\varphi_m$ has length $\bigO{m \log m}$, but only models of size $\geq m$ and extent $\geq 1$.
\end{proof}

\section{Restricted Boolean clones}\label{sec:clones}

Post's lattice of Boolean clones enormously helps to study the
different nature of Boolean functions.
Regarding the propositional satisfiability problem, Lewis showed that the clones
containing $\CloneSE$ are NP-hard, while the
problem is tractable when restricted to arbitrary other clones \cite{lewis79}.

The Boolean clone $\CloneSE$ is the clone of \emph{1-separating}
functions. A function $f(b_1, \ldots, b_n)$ is 1-separating if it has
one argument $b_i$ that is always one if $f$ is one; or equivalently,
if it can be expressed using only the negated implication $\nimp$.
In this section we show that the same dichotomy as above holds for CTL, in the sense that all lower bounds already emerge for the $\CloneSE$ clone.
For the upper bounds of tractable fragments of CTL, see Meier et al.\ \cite{TLPaper}.

In the next lemma, we require the term \emph{short representation}. For a Boolean function $f(\varphi_1,\ldots,\varphi_n)$ to have a \emph{short representation in the base $C$}, it has to be equivalent to a formula $g(\varphi_1,\ldots,\varphi_n)$ using only functions from $C$, with moreover every argument $\varphi_1,\ldots,\varphi_n$ occurring at most once in $g$. For example, $\land$ has a short representation in $\{\neg,\lor\}$ via $\land(\varphi_1,\varphi_2) \equiv \neg(\lor(\neg \varphi_1, \neg \varphi_2))$, whereas $\oplus$ (exclusive or) has none in $\{\land,\lor,\neg\}$.

\begin{lemma}\label{lem:base-translation}
    Let $C$ be a base such that $[C] = \CloneBF$, and let $T \subseteq \TL$. Then every $\varphi \in \B(T)$ has a logspace-constructible, logically equivalent formula $\varphi' \in \B(C,T)$ with $\size{\varphi'} \in \bigO{\size{\varphi}}$.
\end{lemma}
\begin{proof}
In any base $C$ with $[C] = \CloneBF$, the functions $\neg,\land,\lor$ have short representations \cite{lewis79}.
Let $f_\neg(x), f_\land(x,y)$ and $f_\lor(x,y)$ be formulas over $C$ that are short representations of $\neg(x)$, $\land(x,y)$ and $\lor(x,y)$. (Due to commutativity, we can assume that the order in which the arguments appear in $f_\land$ and $f_\lor$ is the same as in $\land$ and $\lor$.)

For $g \in \{\neg,\land,\lor\}$, we define the strings $f_g^p$ (the \emph{prefix} of $f_g$, \ie, the symbols of its body until before its first argument) and $f_g^s$ (the \emph{suffix} of $f_g$, the symbols of its body after its last argument). For $g \in \{\land,\lor\}$, furthermore we define its \emph{middle part} $f_g^m$, \ie, the symbols in $f_g$ between its arguments. For example, $f_\land(y,z)$ can be written down as $f_\land^p \circ y \circ f_\land^m \circ z \circ f_\land^s$, where $\circ$ is the concatenation operation.

Now define $\varphi'$ as a symbol-wise translation of $\varphi$: Any proposition or temporal operator remains unchanged. Any "$g($", for $g \in \{\neg,\land,\lor\}$, is mapped to $f_g^p$. The argument separator "$,$" is mapped to $f_g^m$, where $g$ is the function symbol whose arguments are separated. Finally, any "$)$" is mapped to $f_g^s$, where $g$ is the function symbol whose argument list is closed by "$)$". Since it is possible in logspace to find the corresponding function symbol of a "$,$" or "$)$" (\eg, by going backwards and counting opening and closing parentheses), the whole procedure is implementable in logspace.
\end{proof}

We introduce an equivalence relation between formulas, \emph{frame-equivalence}, that is weaker than logical equivalence but stronger than the equi-satisfiability relation. In particular, this notion also relates the size and extent of satisfying structures.

Two satisfiable formulas $\varphi,\psi$ are called \emph{frame-equivalent} if for every model $(W,R,V,w)$ of $\varphi$ there is a model $(W,R,V',w)$ of $\psi$ (\ie, only the valuations of the propositions are different) and vice versa. Any two equivalent formulas are also frame-equivalent, but in general not the other way around.

\medskip

This notion is used in the next lemma, which shows that certain formulas using the constant function $\top$ have frame-equivalent formulas also without $\top$. This idea is originally due to Lewis, who establish NP-hardness for the $\CloneSE$-fragment of propositional logic, which cannot express $\top$.

Let $\varphi$ be a CTL formula.
$\varphi$ is \emph{non-Boolean} if it is not a proper Boolean combination, \ie, it is a proposition or starts with a CTL operator.
A subformula $\psi \in \SF{\varphi}$ is a \emph{temporal argument} if $\psi$ is directly under the scope of a temporal operator in $\varphi$.
$\varphi$ is now said to be \emph{pseudo-monotone} if $\varphi$, and all temporal arguments $\psi \in \SF{\varphi}$, are Boolean combinations $f(\xi_1,\ldots,\xi_n)$ of non-Boolean formulas in such way that $f$ is monotone in every argument of nonzero temporal depth.
For example, $\AG(\EX \neg p \land \neg q)$ is pseudo-monotone,
but $\AG(\neg \AX p \land \neg q)$ is not (because $\neg \AX p \land \neg q$ is not monotone in the argument $\neg \AX p$).
Similarly, $\neg \EF(\AX p \lor q)$ is not pseudo-monotone, despite all stated formulas being equivalent.

\begin{lemma}\label{lem:remove-top}
Let $C$ be a base such that $\land \in [C]$. Let $k \in \N$ and $T \subseteq \TL$. If $\varphi \in \B_k(C, T)$ is pseudo-monotone, then $\varphi$ has a logspace-constructible, frame-equivalent formula $\psi \in \B_k(C \setminus \{\top\}, T)$ such that $\size{\psi} \in \bigO{\size{\varphi}}$.
\end{lemma}
\begin{proof}
Let $t$ be a proposition that does not occur in $\varphi$. As $\land \in [C]$, the formula $x \land y$ can be written expressed as $f(x,y)$, with $f$ using only functions in $C$.

The formula $\psi$ is now defined as $f(\varphi', t)$, where $\varphi'$ is obtained from $\varphi$ by replacing every occurrence of $\top$ with $t$ and every subformula $QO(\xi)$, for $QO \in T$, with $QO(f(\xi, t))$, and $Q[\xi O \xi']$ with $Q[f(\xi,t)O f(\xi',t)]$.
As the temporal depth is at most $k$, the formula size increases at most by the constant factor $c^{k+1}$, where $c$ depends on the implementation of $f$ (which is not necessarily a short representation). The construction is possible in logspace with a straightforward recursive algorithm that uses only a constant recursion depth.

\smallskip

Every model of $\varphi$ can be converted to a model of $\psi$ by setting $t$ true in every world, as $t$ is then equivalent to $\top$. Conversely, if $\psi$ has a model $\calM$ where $t$ holds in every world, then $\calM$ is a model of $\varphi$.
Consequently, to prove the frame-equivalence of $\varphi$ and $\psi$, we demonstrate that every model $\calM$ of $\psi$ can be enriched to have $t$ labeled in every world.
Formally, given a model $(W,R,V,w)$, we define the valuation $V'$ as $V'(t) \dfn W$, and $V'(p) \dfn V(p)$ for $p \neq t$.
We show then by induction that all subformulas of the form $f(\xi,t) \in \SF{\psi}$ are preserved in all worlds $w \in W$, that is, $(W,R,V,w) \vDash f(\xi,t)$ implies $(W,R,V',w) \vDash f(\xi,t)$.

The induction is on the temporal depth of $\xi$. Let $(W,R,V,w) \vDash f(\xi,t)$. If $\xi \in \calB_0$, then the statement is clearly true. If $\td(\xi) = n > 0$, then $\xi$ is a Boolean combination of non-Boolean formulas $\alpha_1, \ldots, \alpha_k$ such that $\td(\alpha_i) < n$ for all $i$. Every $\alpha_i$ is either a proposition, of the form $QO\beta_i$ or $Q[\beta_i O \gamma_i]$. If $\alpha_i \in \PS$, then $\alpha_i \neq t$, so obviously $V(\alpha_i) = V'(\alpha_i)$. If $\alpha_i$ is of the form $QO\beta_i$ or $Q[\beta_i O \gamma_i]$, then $\beta_i$ and $\gamma_i$ are of the form $f(\xi',t)$. By induction hypothesis, for all $u \in W$, $(W,R,V,u)\vDash \beta_i$ implies $(W,R,V',u) \vDash \beta_i$, and similarly for $\gamma_i$.
By the semantics of the CTL operators, accordingly $(W,R,V,u) \vDash \alpha_i$ implies $(W,R,V',u) \vDash \alpha_i$ for all $u \in W$. Since $\xi$ is monotone in all arguments $\alpha_i \notin \PS$, $(W,R,V',w) \vDash \xi$ and consequently $(W,R,V',w) \vDash f(\xi,t)$ holds.
Since $\psi$ itself is of the form $f(\xi,t)$, the lemma follows.\phantom{$\Box$}
\end{proof}

Lewis's approach in propositional logic, substituting $\top$ with $t$, forces the truth of $t$ by replacing only $\varphi$ itself with $f(\varphi, t)$. For CTL, one could additionally surround the argument $\xi$ of all temporal operators in $\varphi$ with $f(\cdot, t)$. But then the pseudo-monotonicity is still necessary, as the example $\EX \top \land \neg \AX \top$ shows. It is unsatisfiable, but $(\EX(t \land t) \land \neg \AX (t \land t)) \land t$ is satisfiable.

\begin{theorem}\label{thm:s1-frame-reduction}
    Let $C$ be a base such that $\CloneSE \subseteq [C]$. Let $k \in \N$ and $T \subseteq \TL$. Then every $\varphi \in \B_k(T)$ has a logspace-constructible, frame-equivalent formula $\varphi' \in \B_k(C, T)$ such that $\size{\varphi'} \in \bigO{\size{\varphi}}$.
\end{theorem}
\begin{proof}
First, convert $\varphi$, which is over $\{\land,\lor,\neg\}$, to \emph{negation normal form}, \ie, negations $\neg$ appear only in front of propositional variables.
Next, adjoin the constant function $\top$ to the base $C$. From $\CloneSE \subseteq [C]$ it follows $[C \cup \{ \top \}] = \CloneBF$ \cite{Post41}. Consequently, $\varphi$ can be translated to an equivalent formula $\psi \in \B(C\cup\{\top\},T)$ by Lemma~\ref{lem:base-translation}. Since $\varphi$ is in negation normal form, the resulting formula $\psi$ is pseudo-monotone. Conjunction is expressible in $\CloneSE$, \ie, $\land \in [C]$ \cite{Post41}. Therefore we obtain a frame-equivalent formula $\varphi' \in \B_k(C,T)$ by Lemma~\ref{lem:remove-top}.
\end{proof}

It follows from the above result that all lower bounds, with respect to computational complexity or optimal model measures, already hold for any base $C$ that can express $\CloneSE$.

\begin{corollary}\label{cor:hardness-clones}
Let $\CloneSE \subseteq [C]$, $k\in \N$ and $T \subseteq \TL$.
Then $\SAT(\B_k(T)) \leq \SAT(\B_k(C, T))$.
\end{corollary}

\begin{corollary}
\label{thm:model-size-S1}
Let $\CloneSE \subseteq [C]$, $k\in \N$ and $T \subseteq \TL$.
Let $(\varphi_n)_{n \in \N}$ be an infinite family of satisfiable
$\B_k(T)$ formulas such that $\varphi_n$ has minimal model size
$s(n)$ and minimal model extent $e(n)$.
Then there is an infinite family $(\varphi'_n)_{n \in \N}$
of satisfiable $\B_k(C, T)$ formulas with minimal model size
$s(n)$ resp.\ extent $e(n)$, and $\size{\varphi'_n} \in \bigO{\size{\varphi_n}}$.
\end{corollary}

After the lower bounds, the next theorem now generalizes the upper bounds with respect to the standard base $\{\land,\lor\,\neg\}$ to arbitrary bases of Boolean functions, under the condition that the $\AG$ operator is available. The approach is due to Hemaspaandra et~al.\ for a similar result in modal logic \cite{hss10}.

\begin{theorem}\label{thm:base-to-ag}
Let $C$ be a base and $T \subseteq \TL$.
Then every formula $\varphi \in \B(C,T)$ has a logspace-constructible, frame-equivalent formula $\psi \in \B_2(T \cup \{ \AG \})$.
\end{theorem}
\begin{proof}
We transform every $\varphi\in \B(C,T)$ to a formula $\psi \in \B_2(T \cup \{ \AG \})$ such that $\varphi$ and $\psi$ are frame-equivalent. For this we introduce a new atomic proposition $x_\alpha$ for every subformula $\alpha \in \SF{\varphi}$. The idea is that in any model the proposition $x_\alpha$ should be labeled exactly in the worlds where $\alpha$ is true as well.

The formula $\psi$ is defined as $x_\varphi \land \AG \xi$, where
\begin{align*}
\xi \dfn &\bigwedge_{\substack{\alpha \in \SF{\varphi}\\\alpha = f(\beta_1,\ldots,\beta_n)}} \left[x_{\alpha} \leftrightarrow \left(\bigvee_{\substack{\vec b \in \{0,1\}^n\\f(\vec b) = 1}} \bigwedge_{\substack{i \in [n]\\b_i = 1}} x_{\beta_i} \land \bigwedge_{\substack{i \in [n]\\b_i = 0}} \neg x_{\beta_i} \right)\right] \\[1em]
&\land \;\bigwedge_{\mathclap{\alpha \in \SF{\varphi} \cap \PS}} \;\left(x_\alpha \leftrightarrow \alpha\right) \land \bigwedge_{\mathclap{\substack{\alpha \in \SF{\varphi}\\\alpha = QO\beta}}} \left(x_\alpha \leftrightarrow QO x_\beta\right) \land \bigwedge_{\mathclap{\substack{\alpha \in \SF{\varphi}\\\alpha = Q[\beta O\gamma]}}} \left(x_\alpha \leftrightarrow Q[x_\beta O x_\gamma]\right)\text{.}
\end{align*}

Here, $\alpha = f(\beta_1,\ldots,\beta_n)$ means that $\alpha$ is a subformula that starts with a Boolean function $f \in C$ with $\arity{f} = n$. The cases where $\alpha$ is a proposition, or starts with a CTL operator, are handled similarly.

Let $\calK = (W, R, V)$ be a Kripke structure where $\xi$ globally holds.
We prove $(\calK,w) \vDash \alpha \Leftrightarrow (\calK,w)\vDash x_\alpha$ by induction on $\size{\alpha}$ for all $\alpha \in \SF{\varphi}$ and $w \in W$.
If $\alpha \in \PS$, then this is clear. If $\alpha$ starts with a temporal operator, say, $\alpha = QO \beta$, then due to $\xi$ it holds that $x_\alpha$ is true if and only if $QO x_\beta$ is true, which is by induction hypothesis equivalent to $QO \beta$ and hence to $\alpha$. The case of binary temporal operators is similar.
In the case of Boolean functions, the first conjunction in $\xi$ together with the induction hypothesis enforces the correct behaviour; this is easily verified from the definition of semantics of CTL in Section~\ref{sec:prelim}.

For the correctness of the reduction, consider a model $(\calK, w)$ of $\varphi$. For each subformula $\alpha \in \SF{\varphi}$, label $x_\alpha$ in all worlds $w'$ where $(\calK,w') \vDash \alpha$.
Call the resulting model $(\calK^*,w)$.
Then $(\calK^*, w)\vDash x_\varphi$, and again by the CTL semantics, $\xi$ is true in all worlds of $\calK^*$. As a result, $(\calK^*,w)\vDash \psi$.

Conversely, let $(\calK^*,w) \vDash \psi$. We can assume $(\calK^*,w)$ $R$-generable, so $\xi$ globally holds in $\calK^*$.
As a consequence, $(\calK^*,w)\vDash \varphi$ is shown similarly as the other direction.

It remains to show that $\xi$ (and hence $\psi$) is constructible in logarithmic space. Given a formula, it is possible to match parentheses, and consequently to iterate over all subformulas, in logarithmic space using a counter.
Note that each Boolean function $f \in C$ with arity $n$ may have up to $2^n$ satisfying assignments, but for every given base $C$ the maximal arity is constant, hence the large disjunctions have only constantly many disjuncts.
\end{proof}

This result allows to use the polynomial time model checking algorithm of CTL (see Clarke et al.~\cite{clarke_automatic_1986}) on any fragment with the polynomial model property, even under arbitrary bases $C$. Simply translate the formula to a frame-equivalent $\B(\{\land,\lor,\neg\})$ formula first. As the translation has only polynomial blow-up, this preserves the property to have a polynomial model.

\begin{corollary}\label{cor:small-models-c-np}
If $C$ is a base and $\Phi \subseteq \B(C)$ has the polynomial model property, then $\SAT(\Phi) \in \NP$.
\end{corollary}

By Proposition~\ref{prop:model-size-ax}, the $\AX$ fragment of bounded temporal depth has the polynomial model property:

\begin{corollary}\label{cor:ax-in-np}
For all bases $C$ and $k \in \N$, $\SAT(\B_k(C,\AX)) \in \NP$.
\end{corollary}

The same holds for flat CTL due to Theorem~\ref{thm:flat-ctl-upper-bounds}:

\begin{corollary}\label{cor:flat-in-np}
For all bases $C$, $\SAT(\B_1(C,\TL)) \in \NP$.
\end{corollary}

\section{Summary and conclusion}

The results of the previous sections are summarized in the following theorems.
They are also illustrated in Figure~\ref{fig:summary-compl} and \ref{fig:summary-model}.
The first table reproduces all completeness results in a compact way. All $\NP$ lower bounds stem from the propositional satisfiability problem $\SAT(\B_0)$. The $\NP$ upper bounds are all due to a polynomial model property and due to the fact that CTL model checking is in $\P$ \cite{clarke_automatic_1986}. The $\PSPACE$ lower bounds are all due to reduction from the canonical $\PSPACE$-complete problem $\TQBF$, and the upper bounds of $\AG$ and $\AX$ stem from the modal logics $\sfS4\sfD$ and $\sfK\sfD$. The $\{\AX,\AF\}$ fragment, not corresponding to any modal logic, poses an exception; a "pseudo-acyclic" canonical model was constructed for it in Lemma~\ref{lem:balloon-lemma}. Finally, the lower bounds for the $\EXP$-complete cases are shown by a generic reduction from $\APSPACE$, namely for the temporal operators $\AU$, $\A\RLS$, $\{\AG,\AX\}$ and $\{\AG,\AF\}$.

\begin{theorem}\label{thm:hardness-all-S1}
Let $C$ be a base such that $\CloneSE \subseteq [C]$. Let $\emptyset \subsetneq T \subseteq \TL$. Then $\SAT(\B(C, T))$ is
\begin{itemize}
    \item $\PSPACE$-complete if $T = \{ \AX \}$,
  \item logspace-equivalent to $\SAT(\B_2(C,T))$ otherwise, and consequently
\begin{itemize}
  \item $\PSPACE$-complete if $\{\AF \} \subseteq T \subseteq \{\AX, \AF\}$ or $T = \{\AG\}$,
  \item $\EXP$-complete otherwise.
\end{itemize}
\end{itemize}
Furthermore all membership results hold for arbitrary bases $C$.
\end{theorem}
\begin{proof}
The $\PSPACE$ upper bound for $T \subseteq \{\AF,\AX\}$ was shown in Theorem~\ref{thm:afax-in-pspace} for all bases. For $T \subseteq \{\AG\}$ this follows from Propositions~\ref{prop:agef-pspace-complete} and Theorem~\ref{thm:base-to-ag}. The general $\EXP$ upper bound is due to Theorem~\ref{thm:all-c-in-exp} and \ref{thm:base-to-ag}.

The $\AX$ lower bound is due to Theorem~\ref{thm:ax-pspace-completeness}. The hardness for the cases with temporal depth two is due to Theorems~\ref{thm:qbf-to-afeg} and \ref{thm:qbf-to-agef} combined with Corollary~\ref{cor:hardness-clones}. The $\EXP$ lower bounds follow from Theorem~\ref{thm:expc} and Corollary~\ref{cor:hardness-clones}.
\end{proof}

\begin{theorem}\label{thm:ax-c-comp}
Let $C$ a base such that $\CloneSE \subseteq [C]$. Let $k \in \N$. Then the problem $\SAT(\B_k(C,\AX))$ is $\NP$-complete. Furthermore it is in $\NP$ for every base $C$.
\end{theorem}
\begin{proof}
For the standard base $\{\land,\lor,\neg\}$ this is shown in Proposition~\ref{prop:ax-k-npc}. The lower bound therefore follows from Corollary~\ref{cor:hardness-clones}, and the upper bound follows from Corollary~\ref{cor:ax-in-np}.
\end{proof}

\begin{theorem}[Flat CTL]\label{thm:flat-npc}
Let $C$ a base such that $\CloneSE \subseteq [C]$, and $T \subseteq \TL$.
$\SAT(\B_1(C, T))$ is $\NP$-complete. Furthermore it is in $\NP$ for every base $C$.
\end{theorem}
\begin{proof}
Applying Corollary~\ref{cor:hardness-clones}, the $\NP$-hardness already holds for $\SAT(\B_0)$ due to Cook \cite{cook_complexity_1971}. For the upper bound, see Corollary~\ref{cor:flat-in-np}.
\end{proof}

\medskip

Next we present the classification of optimal model measures. It is incomplete for the $\AX$ case with bounded temporal depth, as well as the fragments $\{\AF\}$ and $\{\AF,\AX\}$ of flat CTL.
All other upper and lower bounds are tight.

\begin{theorem}\label{thm:sizes-ctl}
Let $T$ be a non-empty set of temporal operators. Let $C$ be a base such that $\CloneSE \subseteq [C]$.
Let $k \geq 2$.

\begin{enumerate}
    \item If $T = \{\AX\}$, then the optimal model size is  $2^{\bigO{n}} \cap 2^{\bigOmega{\sqrt{n}}}$ for $\B(C,T)$ and $n^{\bigTheta{k}}$ for $\B_k(C,T)$.
    \item For other $T \subseteq \TL$ it is $2^{\bigTheta{n}}$ for $\B(C,T)$ and $\B_k(C,T)$.
    \item $\B(C,T)$ has optimal model extent $\bigTheta{n}$ if $T \in \{\{\AX\},\{\AG\}\}$ and $2^{\bigTheta{n}}$ otherwise. $\B_k(C,T)$ has optimal model extent $k$ if $T=\{\AX\}$, $\bigTheta{n}$ if $T=\{\AG\}$, and again $2^{\bigTheta{n}}$ otherwise.
\end{enumerate}

In the cases of flat CTL holds:

\begin{enumerate}[start=4]
    \item If $T \subseteq \{\AX,\AG\}$, then $\B_1(C,T)$ has optimal model size $\bigTheta{n}$ and extent $\size{T}$.
    \item If $T$ contains $\AU$, $\A\RLS$ or $\{\AG,\AF\}$, then $\B_1(C,T)$ has optimal model size $\bigTheta{n^2}$ and extent $\bigTheta{n}$.
    \item If $T$ contains $\AF$, then $\B_1(C,T)$ has optimal model size at least $\bigOmega{w(n)^2}$ and extent $\bigTheta{n}$.
\end{enumerate}

Furthermore all upper bounds hold for arbitrary bases $C$.
\end{theorem}
\begin{proof}
For flat CTL, all upper and lower bounds stem from Section~\ref{sec:flat-ctl}, namely from Theorem~\ref{thm:minimal-model-size-axag}--\ref{thm:min-size-af}.
All exponential upper bounds follow from Theorem~\ref{thm:upper-size-all}.

The remaining lower bounds follow from Corollary~\ref{cor:model-size-af} and Theorem~\ref{thm:model-size-au-ar} for $\AF$ and $\AU$, Corollary~\ref{cor:ag-model-lower} for $\AG$, Theorem~\ref{thm:model-size-au-ar} for $\A\RLS$ and $\{\AG,\AX\}$, and from Theorem~\ref{theorem:model-size-ax-lower} for $\AX$.
Finally, all lower bounds over $\{\land,\lor,\neg\}$ are transferred to the base $C$ via Theorem~\ref{thm:s1-frame-reduction}.
\end{proof}

\begin{figure}[!ht]\centering\small
\begin{tabular}{llll}\toprule
$T$&
$\B_1$&
$\B_{k \geq 2}$&
$\B$\\
\midrule
$\AX$ & \cellcolor{np}$\NP$-c. & \cellcolor{np}$\NP$-c. &
\cellcolor{pspace}$\PSPACE$-c. \\
$\AG$ & \cellcolor{np}$\NP$-c. & \cellcolor{pspace}$\PSPACE$-c.&
\cellcolor{pspace}$\PSPACE$-c.\\
$\AF\, [, \AX]$ &  \cellcolor{np}$\NP$-c. & \cellcolor{pspace}$\PSPACE$-c.&
\cellcolor{pspace}$\PSPACE$-c. \\
$\AG, \AX, *$ &  \cellcolor{np}$\NP$-c. & \cellcolor{exp}$\EXP$-c.&
\cellcolor{exp}$\EXP$-c.\\
$\AG, \AF, *$ &  \cellcolor{np}$\NP$-c. & \cellcolor{exp}$\EXP$-c.&
\cellcolor{exp}$\EXP$-c.\\
$\AU, *$ & \cellcolor{np}$\NP$-c. & \cellcolor{exp}$\EXP$-c.&
\cellcolor{exp}$\EXP$-c.\\
$\A\RLS, *$  & \cellcolor{np}$\NP$-c. & \cellcolor{exp}$\EXP$-c.&
\cellcolor{exp}$\EXP$-c.\\
\bottomrule
\end{tabular}
\caption{Complexity of $\SAT(\B(T))$ \wrt $\leqlogm$\label{fig:summary-compl}}
\end{figure}
\begin{figure}[!ht]\centering\small
\begin{tabular}{lccccccc}\toprule
$T$&
$\sigma(\B_1)$&$\epsilon(\B_1)$ &
$\sigma(\B_{k \geq 2})$&$\epsilon(\B_{k \geq 2})$ &
$\sigma(\B)$ & $\epsilon(\B)$\\
\midrule
\rule{0pt}{5mm}$\AX$ & \cellcolor{poly}$n$ & \cellcolor{sublin}$1$ &
\cellcolor{poly}$n^\bigTheta{k}$&\cellcolor{sublin}$k$ &
\cellcolor{exp}$\bigOmega{2^{\sqrt{n}}}$, $\bigO{2^n}$ &
\cellcolor{poly}$n$
\\
$\AG$ & \cellcolor{poly}$n$&\cellcolor{sublin}$1$ &
\cellcolor{exp}$2^n$&\cellcolor{poly}$n$&
\cellcolor{exp}$2^n$&\cellcolor{poly}$n$\\
$\AG, \AX$  & \cellcolor{poly}$n$&\cellcolor{sublin}$2$ &
\cellcolor{exp}$2^n$&\cellcolor{exp}$2^n$ &
\cellcolor{exp}$2^n$&\cellcolor{exp}$2^n$\\
\rule{0pt}{5mm}$\AF\, [, \AX]$ & \cellcolor{poly} $\bigOmega{w(n)^2}$, $\bigO{n^2}$ & \cellcolor{poly}$n$
& \cellcolor{exp}$2^n$&\cellcolor{exp}$2^n$ &
\cellcolor{exp}$2^n$&\cellcolor{exp}$2^n$
\\
$\AG, \AF, *$  & \cellcolor{poly}$n^2$&\cellcolor{poly}$n$ &
\cellcolor{exp}$2^n$&\cellcolor{exp}$2^n$&
\cellcolor{exp}$2^n$&\cellcolor{exp}$2^n$\\
$\AU, *$  & \cellcolor{poly}$n^2$&\cellcolor{poly}$n$ &
\cellcolor{exp}$2^n$&\cellcolor{exp}$2^n$ &
\cellcolor{exp}$2^n$&\cellcolor{exp}$2^n$\\
$\A\RLS, *$ & \cellcolor{poly}$n^2$&\cellcolor{poly}$n$ &
\cellcolor{exp}$2^n$&\cellcolor{exp}$2^n$&
\cellcolor{exp}$2^n$&\cellcolor{exp}$2^n$\\
\bottomrule
\end{tabular}

\vspace{3pt}

\tikzcircle[black,fill=sublin]{4pt} constant
\tikzcircle[black,fill=poly]{4pt} polynomial
\tikzcircle[black,fill=exp]{4pt} exponential

\caption{Optimal model size $\sigma$ and extent $\epsilon$ of $\B(T)$, where $n = \bigTheta{\size{\varphi}}$\label{fig:summary-model}}
\end{figure}

\subsection*{Conclusion.}

The results show an interesting property of the computation tree logic
CTL: besides for the pure $\X$ fragment, the computational complexity abruptly
jumps between temporal depth one and two. The flat fragments are all in NP. But already for a nesting depth of two, the complexity of full CTL emerges, which lies between
PSPACE- and EXP-completeness.
This is reasonable if $\AG$ is available, as we then simply can
``pull out'' too deeply nested subformulas until a temporal depth of
only two (see Theorem~\ref{thm:base-to-ag}), but for the other fragments this is still an interesting
result.
From the viewpoint of practical application,
this paper is clearly a negative result, as many important properties
of transition systems are modeled as $\calB_2$- or $\calB_3$-formulas.

When comparing the results to a preceding study for the linear
temporal logic LTL \cite{demri_complexity_2002}, many
similarities arise. All fragments of flat LTL are
NP-complete. LTL also falls down to NP when restricted to one of $\X, \F$ or $\G$;
exponentially long paths cannot be enforced in these cases \cite{SC85}. Here,
the possibility of branching gives an advantage to CTL regarding such
long paths.
On the other hand, the fragments of LTL with PSPACE-complete satisfiability,
namely $\U$ and $\{\F,\G,\X\}$, correspond to the EXP-complete CTL
cases $\AU$, $\{\AG,\AX\}$ and $\{\AG,\AF\}$.

Ultimately, the results for CTL and LTL match very nicely in the sense
that (i) for both logics the bounded $\X$-case is NP-complete and (ii) the
lower bounds for all other operators already hold for temporal depth
of two.

In future research it would be interesting to possibly expand this
principle to similar logics and show similar tight lower bounds.
Candidates would be CTL$^+$, which allows arbitrary Boolean
combinations of temporal operators in the scope of path quantifiers,
then the full branching time logic CTL$^*$ \cite{Emerson86}, and also the
fairness extension of CTL with the operators $\stackrel{\infty}{\F} \dfn \G\F$ and $\stackrel{\infty}{\G} \dfn \F\G$ inside the path
quantifiers.

\section*{Acknowledgments}

The author is thankful to Anselm Haak, Fabian Müller, Arne Meier and Heribert Vollmer from the Institute for Theoretical Computer Science in Hannover for their useful critical comments and for pointing out helpful references, and as well to the anonymous referees for their many valuable corrections and hints.
 
\printbibliography

\end{document}